\documentclass{sig-alternate-05-2015}

\pdfoutput=1

\usepackage{verbatim}
\usepackage{amsmath}
\usepackage{amssymb}
\usepackage{color}
\usepackage{graphicx}
\usepackage{xspace}
\usepackage{bbm}
\usepackage{cite}
\usepackage{hyperref}
\usepackage{listings}             \newcommand{\id}{\ensuremath{\mathbbm{1}}}
\newcommand{\reals}{\ensuremath{\mathbb{R}}}
\newcommand{\naturals}{\ensuremath{\mathbb{N}}}
\newcommand{\prob}{\ensuremath{\mathbf{P}}}
\newcommand{\expect}{\ensuremath{\mathbb{E}}}

\newtheorem{lemma}{Lemma}
\newtheorem{thm}{Theorem}
\newcommand{\catalog}{\mathcal{C}}

\newcommand{\capacity}{\ensuremath{c}}
\newcommand{\source}{\ensuremath{\mathcal{S}}}
\newcommand{\ppath}{\ensuremath{p}}
\newcommand{\requests}{\ensuremath{\mathcal{R}}}
\newcommand{\feasibledomain}{\mathcal{D}}
\newcommand{\argmax}{\mathop{\arg\max\,}}
\newcommand{\trace}{\mathop{\mathtt{trace}}}
\usepackage{algorithm}
\usepackage{algorithmic}

\newcommand{\techrep}[2]{#1}

\newcounter{packednmbr}
\newenvironment{packedenumerate}{\begin{list}{\thepackednmbr.}{\usecounter{packednmbr}\setlength{\itemsep}{0pt}\addtolength{\labelwidth}{-5pt}\setlength{\leftmargin}{\labelwidth}\setlength{\listparindent}{\parindent}\setlength{\parsep}{0pt}\setlength{\topsep}{3pt}}}{\end{list}}
\newenvironment{packeditemize}{\begin{list}{$\bullet$}{\setlength{\itemsep}{0pt}\addtolength{\labelwidth}{-5pt}\setlength{\leftmargin}{\labelwidth}\setlength{\listparindent}{\parindent}\setlength{\parsep}{0pt}\setlength{\topsep}{3pt}}}{\end{list}}

\newcommand{\supp}{\mathtt{supp}}
\newcommand{\CG}{\textsc{MaxCG}\xspace}
\newcommand{\rCG}{\textsc{R-MaxCG}\xspace}
\newcommand{\iCG}{\textsc{I-MaxCG}\xspace}

\begin{document}

\makeatletter
\def\@copyrightspace{\relax}
\makeatother

\clubpenalty=10000 
\widowpenalty = 10000

\title{Adaptive Caching Networks\\ with Optimality Guarantees}

\numberofauthors{1} \author{
\alignauthor
Stratis Ioannidis and Edmund Yeh\\
     \affaddr{E.C.E., Northeastern University}\\
       \affaddr{360 Huntington Ave, 409DA}\\
       \affaddr{Boston, MA, 02115}\\
       \email{\{ioannidis,eyeh\}@ece.neu.edu}
}

\maketitle

\begin{abstract}
We study the problem of optimal content placement over a network of caches, a problem naturally arising in several networking applications, including ICNs, CDNs, and P2P systems. Given  a demand of content request rates and paths followed, we wish to determine the content placement that maximizes the \emph{expected caching gain}, i.e., the reduction of routing costs due to intermediate caching. The offline version of this problem is NP-hard and, in general, the demand and topology may  be a priori unknown. Hence, a distributed, adaptive, constant approximation content placement algorithm is desired.  We show that path replication, a simple algorithm frequently encountered in literature,  can be arbitrarily suboptimal when combined with traditional eviction policies, like LRU, LFU, or FIFO. We propose a distributed, adaptive algorithm that performs stochastic gradient ascent on a concave relaxation of the expected caching gain, and constructs a probabilistic content placement within  $1-1/e$ factor from the optimal, in expectation.  Motivated by our analysis, we also propose a novel greedy eviction policy to be used with path replication, and show through numerical evaluations that both algorithms significantly outperform path replication with traditional eviction policies  over a broad array of network topologies.

 \end{abstract}

\section{Introduction}\label{sec:intro}
We consider a \emph{caching network}, i.e., a network of caches, each capable of storing  a constant number of content items. Certain nodes in the network  act as designated sources for content, and are guaranteed to always store specific items. Any node can generate a request for an item, which is forwarded over a fixed path toward a designated source. However, requests need not reach the end of this path: forwarding stops upon reaching a node that has cached the requested item. Whenever such a ``cache hit'' occurs,  the item is sent over the reverse path towards the node that requested it.

Our goal is  to \emph{allocate items to caches optimally}, i.e.,  in a way that minimizes the aggregate routing costs due to content transfers  across the network. 
This abstract problem   naturally captures---and is directly motivated by---several important real-life networking applications. These include content and information-centric networks (CCNs/ICNs) \cite{jacobson2009networking, rosensweig2013steady,rossi2011caching }, core and edge content delivery networks (CDNs) \cite{borst2010distributed,dehghan2014complexity}, micro/femtocell networks \cite{shanmugam2013femtocaching}, and peer-to-peer networks \cite{lv2002search}, to name a few. For example, in  hierarchical CDNs, requests for content can be served by intermediate caches placed at the network's edge, e.g., within the same administrative domain (e.g., AS or ISP) as the originator of the request; if, however, content is not cached locally, the request can be forwarded to a core server, that acts as a cache of last resort. Similarly, in CCNs, named data items are stored at designated sources, and requests for named content are forwarded to these sources. Intermediate routers can cache items carried by responses, and subsequently serve future requests. Both settings naturally map to the abstract problem we study here. 

In these and many other applications, it is natural to assume that the \emph{demand}, determined by, e.g., how frequently items are requested, and which paths requests follow, is dynamic and not a priori known. For this reason, adaptive algorithms, that (a) discover an optimal item placement without prior knowledge of this demand, and (b) adapt to its changes, are desired. In addition, for large networks, comprising  different administrative domains, collecting information at a single centralized location may be impractical.  Distributed algorithms,   in which a node's caching decisions  rely only on locally available information,  allow the network to scale and are thus preferable.  

A simple, elegant algorithm that attains both properties, and is often encountered in the literature of the different applications mentioned above, is  \emph{path replication}~\cite{cohen2002replication,lv2002search,laoutaris2004meta,rossi2011caching,zhou2004second,che2002hierarchical,jacobson2009networking}. Cast in the context of our problem, the algorithm roughly proceeds as follows: when an item traverses the reverse path towards a node that requested it, it is cached by every intermediate node encountered. When caches are full, evictions  are typicaly implemented using traditional policies, like LRU, LFU, FIFO, etc.

This algorithm is intuitively appealing in its simplicity, and it is clearly both distributed and adaptive to demand. Unfortunately, the resulting allocations of items to caches  come with no guarantees: we show in this paper  that path replication combined with \emph{any} of the above traditional eviction policies is arbitrarily suboptimal. To address this, our main goal is to design a \emph{distributed, adaptive caching algorithm with provable performance guarantees}. To that end, we make the following contributions:
\begin{packeditemize}
\item We set the problem of optimal caching network design on a formal foundation. We do so by rigorously defining the problem of finding an \emph{allocation}, i.e., a mapping of items to network caches, that maximizes the \emph{expected caching gain}, i.e., the routing cost reduction achieved due to caching at intermediate nodes. The deterministic, combinatorial version of the problem is NP-hard, though it is approximable within a $1-1/e$ factor \cite{shanmugam2013femtocaching,ageev2004pipage}.
\item We prove that the classic path replication algorithm, combined with LRU, LFU, or FIFO eviction policies, leads to allocations that are \emph{arbitrarily suboptimal.} Our result extends to any  \emph{myopic} strategy, that ignores costs incurred upstream due to cache misses.
\item We construct a distributed, adaptive algorithm that converges to a probabilistic allocation of items to caches that is  within a $1-1/e$ factor from the optimal, \emph{without} prior knowledge of the demand (i.e., items requested and routes followed) or the network's topology. The algorithm performs a projected gradient ascent over a concave objective approximating the expected caching gain. 
\item Motivated by this construction, we also propose a new eviction policy to be used  with path replication: whenever an item is back-propagated over a path, the nodes on the path have the opportunity to store it and evict an existing content, according to a greedy policy we design.
\item We show through extensive simulations over a broad array of both synthetic and real-life topologies that both algorithms significantly outperform traditional eviction policies. In all cases studied, the greedy heuristic performs exceptionally well, achieving at least 95\% of the gain achievable by the projected gradient ascent algorithm, that comes with provable guarantees.
\end{packeditemize}
Our analysis requires overcoming several technical hurdles. To begin with, constructing our distributed algorithm, we show that it is always possible to construct a probabilistic allocation, mapping items to caches, that satisfies capacity constraints \emph{exactly}, from a probabilistic allocation that satisfies capacity constraints only in \emph{expectation}. Our construction, which is interesting in its own right, is simple and intuitive, and can be performed in polynomial time. Moreover, the concave relaxation we study is non-differentiable; this introduces additional technical difficulties when performing projected gradient ascent, which we address.

The remainder of this paper is structured as follows. We review related work in Section~\ref{sec:related}. We formally introduce our problem  in Section~\ref{sec:model}, and discuss offline algorithms for it solution in Section~\ref{sec:pipage}. Our main results on distributed, adaptive algorithms in Section~\ref{sec:adaptive}; we also prove an equivalence theorem on several variants of the expected caching gain maximization problem in Section~\ref{sec:equivalence}. Finally, Section~\ref{sec:numerical} contains our evaluations,  and we conclude in Section~\ref{sec:conclusions}.

\section{Related Work}\label{sec:related}
Path replication is best known as the de facto caching mechanism in content-centric networking \cite{jacobson2009networking}, but has a long history in  networking literature. In their seminal paper, Cohen and Shenker~\cite{cohen2002replication} show that path replication, combined with constant rate of evictions leads to an allocation that is optimal, in equilibrium, when nodes are visited through uniform sampling. This is one of the few results on path replication's optimality (see also \cite{ioannidis2009absence}); our work (c.f., Theorem~\ref{subopt}) proves that, unfortunately, this result does not generalize to routing over arbitrary topologies. Many studies provide numerical evaluations of path replication combined with simple eviction policies, like LRU, LFU, etc., over different topologies (see, e.g., \cite{laoutaris2004meta,rossi2011caching,zhou2004second}). In the context of  CDNs and ICNs, Rosensweig et al.~\cite{rosensweig2013steady} study conditions under which path replication with LRU, FIFO, and other variants, under fixed paths, lead to an ergodic chain. Che et al. \cite{che2002hierarchical}  approximate the LRU policy hit probability through a TTL-based eviction scheme; this approach that has been refined and extended in several recent works to model many traditional eviction policies \cite{fricker2012versatile,martina2014unified,berger2014exact,fofack2012analysis}; alternative analytical models are explored in \cite{gallo2012performance,carofiglio2011modeling}. None of the above works however study optimality issues or guarantees. 

Several papers have studied complexity and optimization issues in offline caching problems \cite{baev2008approximation,bartal1995competitive,fleischer2006tight,shanmugam2013femtocaching,applegate2010optimal,borst2010distributed}. With the exception of \cite{shanmugam2013femtocaching}, these works model the network as a bipartite graph:  nodes generating requests connect directly to caches, and demands are satisfied a single hop (if at all). 
Beyond content placement, Deghan et al.~\cite{dehghan2014complexity} jointly optimize caching and routing in this bipartite setting.
In general, the \emph{pipage rounding} technique of Ageev and Sviridenko~\cite{ageev2004pipage} (see also \cite{calinescu2007maximizing,vondrak2008optimal}) yields again a constant approximation algorithm in the bipartite setting, while approximation algorithms are also known for several variants of this problem~\cite{baev2008approximation,bartal1995competitive,fleischer2006tight,borst2010distributed}. 
 
Among these papers on offline caching, the recent paper by Shanmugam et al.~\cite{shanmugam2013femtocaching} is closest to the problem we tackle here; we rely and expand upon this work.  Shanmugam et al.~consider  wireless  nodes that download items from (micro/femtocel) base stations in their vicinity. Base stations are visited in a predefined order (e.g., in decreasing order of connection quality), with the wireless service acting as a ``cache of last resort''. This can be cast as an instance of our problem, with paths defined by the traversal sequence of base-stations, and the network graph defined as their union.  
The authors show that determining the optimal allocation is NP-hard, and that an $1-1/e$ approximation algorithm can be obtained through pipage rounding; we review these results, framed in the context of our problem, in Section~\ref{sec:pipage}.

All of the above complexity papers \cite{baev2008approximation,bartal1995competitive,fleischer2006tight,ageev2004pipage}, including \cite{shanmugam2013femtocaching}, study \emph{offline}, \emph{centralized} versions of their respective caching problems. Instead, we focus on providing \emph{adaptive, distributed} algorithms, that operate without any prior knowledge of the demand or topology. In doing so, we we produce a distributed algorithms for (a) performing projected gradient ascent over the concave objective used in pipage rounding, and (b) rounding the objective across nodes; combined, these lead to a distributed, adaptive caching algorithm with provable guarantees~(Thm.~\ref{maincor}). 

Adaptive replication schemes exist for asymptotically large, single-hop CDNs, \cite{leconte2015designing, leconte2012bipartite,ioannidis2010distributed},  but these works do not explicitly model a graph structure. The dynamics of the greedy path replication algorithm we propose in Section~\ref{sec:greedy} resemble the greedy algorithm used to make caching decisions in~\cite{ioannidis2010distributed}, though our objective is different, and we cannot rely on a mean-field approximation in our argument. The dynamics are also similar (but not identical) to the dynamics of the ``continuous-greedy'' algorithm used for submodular maximization~\cite{vondrak2008optimal} and the Frank-Wolfe algorithm~\cite{clarkson2010coresets}; these can potentially serve as a basis for formally establishing its convergence, which we leave as future work.

The path replication eviction policy we propose also relates to greedy maximization techniques used in throughput-optimal backpressure algorithms---see, e.g., Stolyar~\cite{stolyar2005maximizing} and, more recently, Yeh et al.~\cite{yeh2014vip}, for an  application to throughput-optimal caching in ICN networks. We minimize routing costs and  ignore throughput issues, as we do not model congestion. Investigating how to combine these two research directions, capitalizing on commonalities between these greedy algorithms, is an interesting open problem.

\section{Model}\label{sec:model}
We consider a network of caches, each capable of storing at most a constant number of content items. Item requests are routed over given (i.e., fixed) routes, and are satisfied upon hitting the first cache that contains the requested item. Our goal is to determine an item allocation (or, equivalently, the contents of each cache), that minimizes the aggregate routing cost.  We describe  our model in detail below.

\subsection{Cache Contents and Designated Sources}
We represent a network as a directed graph $G(V,E)$.  
Content items (e.g., files, or file chunks) of equal size   are to be distributed across network nodes. In particular, each node is associated with a cache, that can store a finite number of items. We denote by  $\catalog$ the set of content items available, i.e., the \emph{catalog}, and assume that $G$ is symmetric, i.e.,  $(i,j)\in E$ if and only if $(j,i)\in E$. 

We denote by $\capacity_v\in \naturals$ the cache capacity at node $v\in V$: exactly $\capacity_v$ content items are stored in this node. 
For each node $v\in V$, we denote by $$x_{vi}  \in \{0,1\},\quad \text{ for } v\in V, i \in \catalog, $$
the variable indicating whether $v$ stores content item $i$. We denote by $X=[x_{vi}]_{v\in V,i \in \catalog}\in \{0,1\}^{|V|\times |\catalog|}$ the matrix whose rows comprise the indicator variables of each node. We refer to $X$ as the \emph{global allocation strategy} or, simply, \emph{allocation}.
Note that the capacity constraints imply that
$$\textstyle\sum_{i\in \catalog} x_{vi} = c_v, \quad\text{for all }v\in V.$$

We associate each item $i$ in the catalog $\catalog$  with a fixed set of \emph{designated sources} $\source_i\subseteq V$, that always store $i$. That is:
$$x_{vi} = 1 \text{, for all } v\in\source_i.  $$
Without loss of generality, we assume that the sets $\source_i$ are feasible, i.e.,
$\sum_{i: v\in \source_i} x_{vi} \leq c_v,$ for all $v\in V.$

\subsection{Content Requests and Routing Costs}
The network serves content requests  routed over the graph $G$.  A request is determined by (a) the  item requested, and (b) the path that the request follows. Formally, a \emph{path} $\ppath$ of length $|p|=K$ is a sequence $$\{p_1,p_2,\ldots,p_K\}$$ of nodes  $p_k\in V$ such that edge $(p_k,p_{k+1})$ is in $E$, for every $k\in \{1,\ldots,|p|-1\}$. Under this notation, a \emph{request} $r$ is a pair $(i,p)$ where $i\in \catalog$ is the item requested, and $\ppath$ is the path traversed to serve this request.
 We say  that a request $(i,p)$  is \emph{well-routed} if the following  natural assumptions hold: 
\begin{packeditemize}
\item[(a)] The path $p$ is simple, i.e., it contains no loops.
\item[(b)] The terminal node in the path  is a designated source node for $i$, i.e., if $|p|=K$, 
$p_K \in S_i.$
\item[(c)] No other node in the path is a designated source node for $i$, i.e., if $|p|=K$,
$p_k\notin S_i$, for $k=1,\ldots,K-1.$ 
\end{packeditemize}
 We denote by $\requests$  the set  of all requests. Without loss of generality, we henceforth assume that all requests in $\requests$ are well-routed.  Moreover, requests for each element $\requests$ arrive according to independent Poisson processes; we denote by $\lambda_{(i,p)}>0$ the arrival rate of a request $(i,p)\in \requests$.

An incoming request $(i,p)$ is routed over the network $G$  following path $p$, until it reaches a cache that stores $i$. At that point, a \emph{response} message is generated, carrying the item requested. The response is propagated over $p$ in the reverse direction, i.e., from the node where the ``cache hit'' occurred, back to the first node in $p$, from which the request originated. 
To capture  costs  (e.g., delay, money, etc.), we associate a \emph{weight} $w_{ij}\geq 0$ with each  edge $(i,j)\in E$, representing the cost of transferring an item across this edge.
 We assume  that (a) costs are solely due to response messages that carry an item, while request forwarding costs are negligible, and (b) requests and downloads are instantaneous (or, occur at a smaller timescale compared to the request arrival process).  We do not assume that $w_{ij}=w_{ji}$.

When a request $(i,p)\in\requests$ is well-routed, the cost for serving it can be written concisely in terms of the  allocation:
\begin{align}C_{(i,p)} = C_{(i,p)}(X) = \sum_{k=1}^{|p|-1} w_{p_{k+1}p_k} \prod_{k'=1}^k (1-x_{p_{k'}i}).\label{detcost}\end{align}
Intuitively, this formula states that $C_{(i,p)}$ includes the cost of an edge $(p_{k+1},p_k)$ in the path $p$ if \emph{all} caches preceding this edge in $p$ do not store $i$. If the request is well-routed, no edge (or cache) appears twice in \eqref{detcost}. Moreover, the last cache in $p$ stores the item, so the request is always served.

\subsection{Maximizing the Caching Gain}

As usual, we seek an allocation that minimizes the aggregate expected cost. In particular, let $C_0$ be the expected cost per request, when requests are served by the designated sources at the end of each path, i.e., 
\begin{align*}
C_0 =   \textstyle\sum_{(i,p) \in \requests }\lambda_{(i,p)}  \sum_{k=1}^{|p|-1} w_{p_{k+1}p_k}.
\end{align*}
Since requests are well-routed, $C_0$ is an upper bound on the expected routing cost. 
Our objective is to determine a feasible allocation $X$ that maximizes the \emph{caching gain}, i.e., the expected cost reduction attained due to caching at intermediate nodes, defined as:
\begin{align}
\begin{split}
F(X)&=  C_0-\textstyle\sum_{(i,p) \in \requests }\lambda_{(i,p)}  C_{(i,p)}(X) \\
&= \sum_{(i,p) \in \requests }\!\!\lambda_{(i,p)}\!\!  \sum_{k=1}^{|p|-1} \!\!w_{p_{k+1}p_k}\left(1\!-\! \prod_{k'=1}^k (1\!-\!x_{p_{k'}i}) \right) 
\end{split}\label{gain}
\end{align}
In particular, we seek solutions to  the following problem:
\begin{subequations}\label{deterministic}
\center{\CG}\vspace*{-0.5em}
\begin{align}
\text{Maximize:}& \quad F(X) \\
\text{subj.~to:}& \quad X\in \feasibledomain_1
\end{align}
\end{subequations}
where 
$\feasibledomain_1$ is the set of matrices $X\in  \reals^{|V|\times |\catalog|}$ satisfying the capacity, integrality, and source constraints, i.e.:
\begin{subequations}\label{integralconstr}
\begin{align}
& \textstyle\sum_{i\in\catalog} x_{vi} = \capacity_v, & &\text{ for all }v\in V \label{xcapacity}\\
&x_{vi}\in \{0,1\}, &&\text{ for all }v\in V, i\in\catalog, \text{ and} \label{xintegrality}\\
&x_{vi} = 1, &&\text{ for all }v\in \source_i\text{ and all }i\in\catalog. \label{xsource}   	
\end{align}
\end{subequations}
Problem \CG is NP-hard (see Shanmugam et al.~\cite{shanmugam2013femtocaching} for a reduction from the 2-Disjoint Set Cover Problem). Our objective is to solve \CG using a \emph{distributed}, \emph{adaptive} algorithm, that produces an allocation within a constant approximation of the optimal, \emph{without} prior knowledge of the network topology, edge weights, or the demand.

\begin{table}[!t]
\caption{Notation Summary}
\begin{small}
\begin{tabular}{cl}
\hline
$G(V,E)$ & Network graph, with nodes $V$ and edges $E$\\
$\catalog$ & Item catalog\\
$c_v$ & Cache capacity at node $c\in V$\\
$w_{uv}$ & Weight of edge $(u,v)\in E$\\
$\requests$ & Set of requests $(i,p)$, with $i\in\catalog$ and $p$ a path\\
$\lambda_{(i,p)}$ & Rate of request $(i,p)\in \requests$ \\
$x_{vi}$ & Variable indicating $v\in V$ stores $i\in\catalog$ \\
$y_{vi}$ & Marginal probability that $v\in V$ stores $i\in\catalog$\\
$X$ & $|V|\times |\catalog|$ matrix of $x_{vi}$, for $v\in V$, $i \in \catalog$\\
$Y$ &  $|V|\times |\catalog|$ matrix of marginals $y_{vi}$, $v\in V$, $i \in \catalog$ \\
$\feasibledomain_1$ & Set of feasible allocations $X\in \{0,1\}^{|V|\times |\catalog|}$\\
$\feasibledomain_2$ & Convex hull of $\feasibledomain_1$\\
$F$ & The expected caching gain \eqref{gain} in $\feasibledomain_1$, \\&and its multi-linear relaxation \eqref{equality} in $\feasibledomain_2$\\
$L$ & The concave approximation \eqref{relaxedobjective} of $F$\\	
\hline
\end{tabular}
\end{small}
\end{table}

\section{Pipage Rounding}\label{sec:pipage}
Before presenting our distributed, adaptive algorithm for solving \CG, we first discuss how to obtain a constant approximation solution in polynomial time in a \emph{centralized}, \emph{offline} fashion. To begin with, \CG is a submodular maximization problem under matroid constraints: hence, a solution within a 1/2 approximation from the optimal can be constructed by a greedy algorithm.\footnote{Starting from items placed only at designated sources, this algorithm iteratively adds items to caches, selecting at each step a feasible assignment $x_{vi}=1$ that leads to the largest increase in the caching gain.} The solution we present below, due to Shanmugam et al.~\cite{shanmugam2013femtocaching}, improves upon this ratio using  a technique called \emph{pipage rounding}~\cite{ageev2004pipage}. 
In short, the resulting approximation algorithm consists of two steps: (a) a  \emph{convex relaxation} step, that relaxes the integer program to a convex optimization problem, whose solution is within a constant approximation from the optimal, and (b) a \emph{rounding} step, in which the (possibly) fractional solution is rounded to produce a solution to the original integer program. The convex relaxation plays an important role in our distributed, adaptive algorithm; as a result, we briefly overview pipage rounding as applied to  \CG below, referring the interested reader to \cite{shanmugam2013femtocaching,ageev2004pipage} for further details.

\noindent\textbf{Convex Relaxation}. To construct a convex relaxation of \CG, suppose that variables $x_{vi}$, $v\in V$, $i\in \catalog$, are independent Bernoulli random variables.  Let $\nu$ be the corresponding joint probability distribution defined over matrices in $\{0,1\}^{|V|\times|\catalog|}$, and denote by $\prob_{\nu}[\cdot]$, $\expect_{\nu}[\cdot]$ the probability and expectation w.r.t.~$\nu$, respectively. Let  $y_{vi}$, $v\in V$, $i\in \catalog$, be the (marginal) probability that $v$ stores $i$, i.e.,
\begin{align}\label{marginals}y_{vi} = \prob_\nu[x_{vi}=1] = \expect_\nu[x_{vi}].\end{align}
Denote by  $Y=[y_{vi}]_{v\in V,i\in\catalog}\in  \reals^{|V|\times |\catalog|}$  the matrix comprising the marginal probabilities \eqref{marginals}.  
Then, for $F$  given by~\eqref{gain}:
\begin{align}
\expect_\nu[F(X)]& = \!\! \sum_{(i,p) \in \requests }\!\!\lambda_{(i,p)}\!\!  \sum_{k=1}^{|p|-1}\!\! w_{p_{k+1}p_{k}}\!\!\left(1\!-\!\expect_\nu \left [ \prod_{k'=1}^k (1\!-\!x_{p_{k'}i}) \right]\right)\nonumber\\
& \stackrel{}{=} \!\!\sum_{(i,p) \in \requests }\!\!\lambda_{(i,p)}\!\!  \sum_{k=1}^{|p|-1}\!\! w_{p_{k+1}p_{k}}\!\!\left(1\!-\! \!\prod_{k'=1}^k \!\! \left(1\!-\!\expect_\nu\left [x_{p_{k'}i}\right]\right)\right) \nonumber \\
& = F(Y).\label{equality}
\end{align}
Note that the second equality holds by independence, and the fact that path $p$ is simple (no node appears twice). This extension of $F$ to the domain $[0,1]^{|V|\times|\catalog|}$ is known as the \emph{multi-linear relaxation} of $F$. Consider now the problem:
\begin{subequations}\label{indep}
\begin{align}
\text{Maximize:}& \quad F(Y)\\
\text{subject to:}& \quad Y\in \feasibledomain_2,
\end{align}
\end{subequations}
where
$\feasibledomain_2$ is the set of matrices $Y=[y_{vi}]_{v\in V,i\in\catalog}\in  \reals^{|V|\times |\catalog|}$ satisfying the capacity and source constraints, with the integrality constraints relaxed, i.e.:
\begin{subequations}\label{relax}
\begin{align}
& \textstyle\sum_{i\in\catalog} y_{vi} = \capacity_v, &&\text{ for all }v\in V\\
&y_{vi}\in [0,1],&&\text{ for all }v\in V, i\in\catalog, \text{ and}\\
&y_{vi} = 1, &&\text{ for all }v\in \source_i\text{ and all }i\in\catalog.   	
\end{align}
\end{subequations}
Note that an allocation $X$ sampled from  a $\nu$ with marginals  $Y\in\feasibledomain_2$ only satisfies the capacity constraints \emph{in expectation}; hence, $X$ may not be in $\feasibledomain_1$.  Moreover, if $X^*$ and $Y^*$ are optimal solutions to \eqref{deterministic} and  \eqref{indep}, respectively, then
\begin{align}F(Y^*)\geq F(X^*),\label{trivial}\end{align}
as \eqref{indep} maximizes the same function over a larger domain.

The multi-linear relaxation \eqref{relax} is not concave, so \eqref{indep} is \emph{not} a convex optimization problem.
 Nonetheless,  \eqref{indep} can  be approximated as follows. Define $L:\feasibledomain_2\to \reals$ as: 
\begin{align}\label{relaxedobjective}
L(Y) =  \sum_{(i,p) \in \requests }\lambda_{(i,p)}  \sum_{k=1}^{|p|-1} w_{p_{k+1}p_k} \min \{1,\sum_{k'=1}^k y_{p_{k'}i}\}.
\end{align}
Note  that $L$ is concave, and consider now the problem:
\begin{subequations}\label{convex}
\begin{align}
\text{Maximize:} & \quad L(Y)\\
\text{Subject to:} &\quad Y\in \feasibledomain_2.
\end{align} 
\end{subequations}
Then, the optimal value of \eqref{convex} is guaranteed to be within a constant factor from the optimal value of \eqref{indep}--and, by \eqref{trivial}, from the optimal value of \eqref{deterministic} as well. In particular: 
\begin{thm}\cite{ageev2004pipage,shanmugam2013femtocaching}\label{approximation}
Let $Y^*$, and $Y^{**}$ be optimal solutions to  \eqref{indep} and \eqref{convex}, respectively. Then,
\begin{align} F(Y^*)\geq  F(Y^{**}) \geq (1-\frac{1}{e}) F(Y^*). \label{sandwitch}\end{align}
\end{thm}
\techrep{For completeness, we prove this in Appendix~\ref{proofofapprox}.}{A proof can be found in \cite{arxiv}.} Problem \eqref{convex} is convex; in fact, by introducing auxiliary variables, it can be converted to a linear program and, as such, the minimizer $Y^{**}$ can be computed in strongly polynomial time \cite{ageev2004pipage}. 
\noindent\textbf{Rounding.} To produce a constant approximation solution to \CG, the solution $Y^{**}$ of \eqref{convex} is rounded. The rounding scheme is based on the following property of $F$: given a fractional solution $Y\in \feasibledomain_2$, there is always a way to transfer mass between any two fractional variables $y_{vi}$, $y_{vi'}$ so that\footnote{Properties (a)-(c)  are a direct consequence of the convexity of $F$ when restricted to any two such variables, a property that Ageev and Sviridenko refer to as \emph{$\epsilon$-convexity} \cite{ageev2004pipage}.} (a) at least one of them becomes 0 or 1, (b) the resulting $Y'$ under this transformation is feasible, i.e., $Y'\in \feasibledomain_2$, and (c) the caching gain at $Y'$ is at least as good as at $Y$, i.e.,  $F(Y')\geq F(Y)$. 

This suggests the following iterative algorithm for producing an integral solution. 
\begin{packedenumerate}
\item Start from  $Y^{**}$, an optimal solution  to the problem \eqref{convex}.
\item  If the solution is fractional, find two variables  $y_{vi}$, $y_{vi'}$ that are fractional: as capacities are integral, if a fractional variable exists, then there must be at least two.
\item Use the rounding described by properties (a)-(c) to transform (at least) one of these two variables to either 0 or 1, while increasing the caching gain $F$.
\item Repeat steps 2-3 until there are no fractional variables.
\end{packedenumerate}
 As each  rounding step reduces the number of fractional variables by at least 1, the above algorithm concludes in at most $|V|\times |\catalog|$ steps, producing an integral solution $X'\in \mathcal{D}_1$. Since each rounding step can only increase $F$, $X'$ satisfies:
$$F(X')\geq F(Y^{**}) \stackrel{\eqref{sandwitch}}{\geq} (1-\frac{1}{e})F(Y^*) \stackrel{\eqref{trivial}}{\geq} (1-\frac{1}{e}) F(X^*),$$
i.e., is  a $(1-\frac{1}{e})$-approximate solution to \CG.

\section{Distributed Adaptive Caching}\label{sec:adaptive}

\subsection{ Path Replication Suboptimality }\label{sec:subopt}

Having discussed how to sove \CG offline, we  turn our attention to distributed, adaptive algorithms. We begin with a negative result: the simple path replication algorithm described in the introduction, combined with  LRU, LFU, or FIFO evictions, is arbitrarily suboptimal.

We prove this below using the simple star network illustrated in Figure~\ref{example}, in which only one file can be cached at the central node. Intuitively, when serving requests from the bottom node,  path replication with, e.g., LRU evictions, alternates between storing either of the two files. However,  the optimal allocation is to permanently store file 2, (i.e., $x_{v2}=1$): for large $M$, this allocation leads to a caching gain arbitrarily larger than the one under path replication and LRU. As $c_v=1$,   LFU and FIFO coincide with LRU, so the result extends to these policies as well.

Formally, assume that request traffic is generated only by node $u$:  requests for items 1 and 2 are routed through paths $p_1=\{u,v,s_1\}$ and $p_2=\{u,v,s_2\}$  passing through node $v$, that has a cache of capacity $c_v=1$. Let  $\lambda_{(1,p_1)}=1-\alpha$, $\lambda_{(2,p_2)}=\alpha$, for $\alpha\in(0,1)$.
As illustrated in Figure~\ref{example}, the routing cost  over edges $(s_1,u)$ and $(v,u)$ is 1, and the routing cost over $(s_2,u)$ is $M\gg 1$. 
Then, the following holds:

\begin{thm} \label{subopt}Let $X(t)\in\{0,1\}^2$ be the allocation of the network in Figure~\ref{example} at time $t$, under path replication with an LRU, LFU, or FIFO policy. Then, for $\alpha=\frac{1}{\sqrt{M}}$,
$$\lim_{t\to\infty } \expect[F(X(t))]/\max_{X\in \feasibledomain_1}F(X)  = O({1}/{\sqrt{M}}) .$$
\end{thm}
\begin{figure}[!t]
\includegraphics[width=\columnwidth]{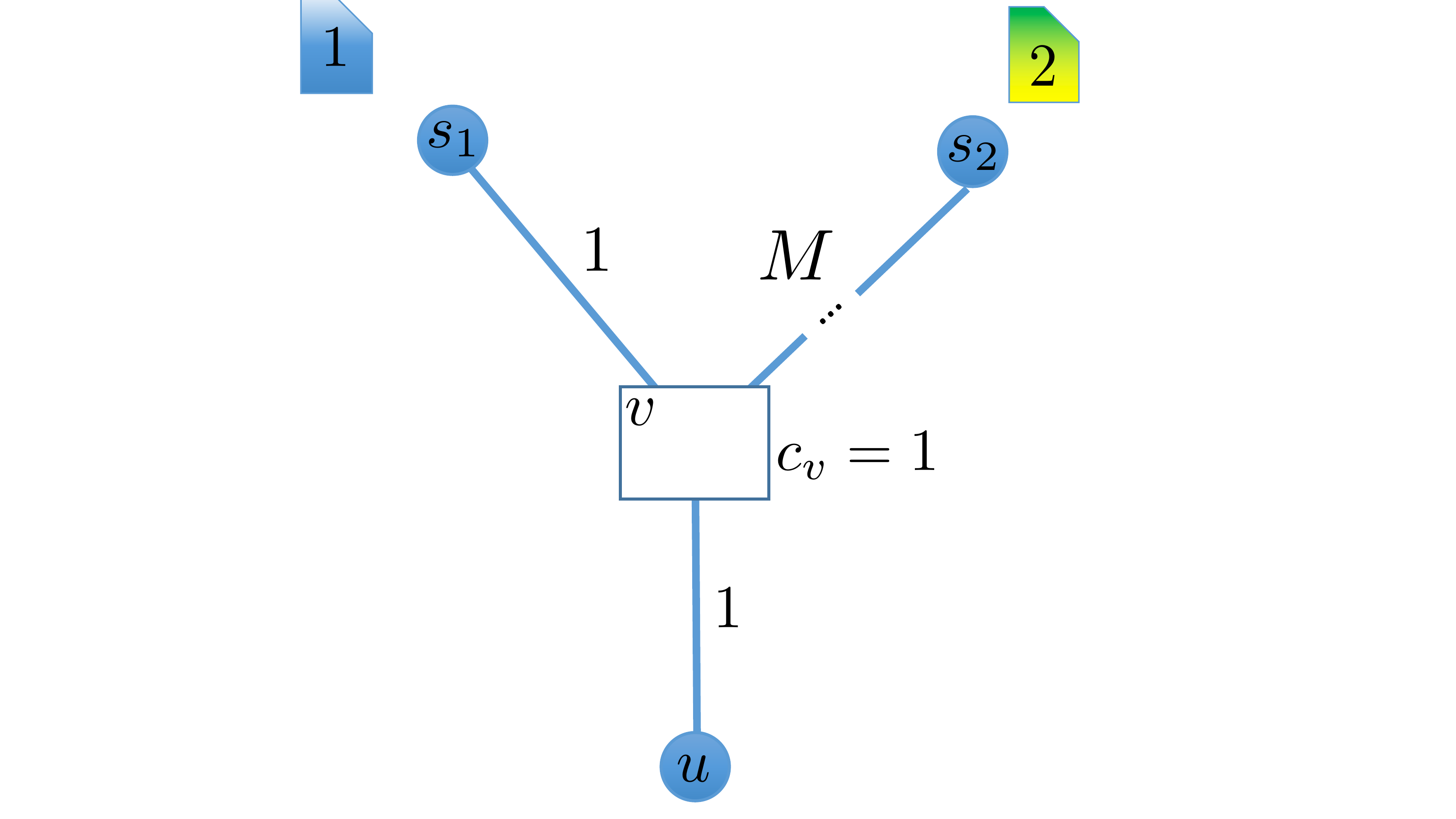}
\vspace*{-0.5cm}
\caption{A simple caching network, with $\catalog=\{1,2\}$, $S_1=\{s_1\}$, $S_2=\{s_2\}$. The cache at $v$ has capacity $c_v=1$, and the cost of the edge between $v$ and $s_2$ is $M\gg1$. Node $u$ requests item 1 with rate $1-\alpha$, and item 2 with rate $\alpha$. For $\alpha=\frac{1}{\sqrt{M}}$, path replication with LRU, LFU, or FIFO leads to an arbitrarily suboptimal caching allocation, in steady state.}\label{example}
\end{figure}
\begin{proof}
The worst case cost, when cache $v$ is empty, is:
$C_0 = \alpha \times(M+1) + (1-\alpha)\times 2 = \alpha M +2 - \alpha.  $
Suppose that $u$ permanently caches item 2. This results in an expected routing cost of $\alpha \times 1 + (1-\alpha) \times 2 = 2-\alpha.$ Hence, an optimal allocation $X^*$ necessarily has a caching gain $F(X^*) >  \alpha M +2-\alpha -(2-\alpha) = \alpha M.$
Consider now the path replication algorithm, in which either item is cached at $v$ whenever it is back-propagated over the reverse path. As $c_v=1$, the LRU, FIFO and LFU policies coincide, and yield exactly the same eviction decision. Moreover, as request arrivals are independent Poisson, the steady state probabilities that $v$ stores item 1 or 2 are  $1-\alpha$ and $\alpha$, respectively. Hence, the expected routing cost in steady state is
$\alpha^2 +\alpha(1-\alpha)\times (M+1)  + (1-\alpha)\alpha\times 2 + (1-\alpha)^2 =  1+\alpha M-\alpha^2M +\alpha -\alpha^2),$
leading to a caching gain of 
$\alpha M +2 -\alpha - (1+\alpha M - \alpha^2 M +\alpha-\alpha^2) = \alpha^2 M +1- 2\alpha +\alpha^2.$
Hence, the ratio of the expected caching gain under path replication with LRU, LFU, or FIFO evictions to $F(X^*)$ is at most
$\alpha+ \frac{(1- \alpha)^2}{\alpha M},$
and the theorem follows for $\alpha=1/\sqrt{M}$.
\end{proof}
 Taking $M$ to be arbitrarily large therefore makes this ratio arbitrarily small. Clearly, this argument  applies to any eviction strategy that is \emph{myopic}, \emph{i.e.}, is insensitive to upstream costs. Accounting for the cost of an item's retrieval when caching seems necessary to provide any optimality guarantee; this is the case for the algorithms we propose below.

\subsection{Projected Gradient Ascent}\label{sec:pga}
Given the negative result of Theorem~\ref{subopt}, we now describe our distributed, adaptive algorithm for solving \CG. Intuitively, the algorithm performs a \emph{projected gradient ascent}  over  function $L$, effectively solving the convex problem \eqref{convex} in a distributed, adaptive fashion.  The concavity of $L$ ensures convergence (in contrast to minimizing $F$ directly), while Theorem~\ref{approximation} ensures that the caching gain attained in steady state is within an $1-\tfrac{1}{e}\approx 0.62$ factor from the optimal. We describe the algorithm in detail below. 

We deal with the following two challenges.
First, in each step of  gradient ascent, a node must estimate the contribution of its own caching allocation to the gradient of the (global) function $L$. The estimation should rely only on local information; additional care needs to be taken as $L$ is not differentiable in the entire domain $\feasibledomain_2$, so a \emph{subgradient} needs to be estimated instead.
Second, the final value of the convex relaxation \eqref{convex}, as well as intermediate solutions during gradient ascent, produce fractional values $Y\in\feasibledomain_2$.  To determine what to place in each cache, discrete allocations $X\in\feasibledomain_1$ need to be determined from $Y$. Our algorithm \emph{cannot} rely on pipage rounding to construct such cache allocations: each node must determine its cache contents in a distributed way, without explicitly computing $F$. 

We address both challenges, by proving  that (a) a \emph{feasible} randomized rounding of any $Y\in \feasibledomain_2$, and (b) the subgradient of $L$ w.r.t.~$Y$ can both be computed in a distributed fashion, using only information locally available at each node.

\begin{algorithm}[t]
\begin{small}
  \caption{\textsc{Projected Gradient Ascent}}\label{alg:ascent}
    \begin{algorithmic}[1]
   \STATE Execute the following at each $v\in V$:
   \STATE Pick arbitrary state $y_{v}^{(0)}\in \feasibledomain_2^v$.
   \FOR{ each period $k\geq 1$ and each $v\in V$} 
    \STATE Compute the sliding average $\bar{y}_v^{(k)}$ through \eqref{slide}.
    \STATE Sample a $x^{(k)}_v\in \feasibledomain_1^v$ from a $\mu_v$ that satisfies \eqref{sample}.
    \STATE Place items $x^{(k)}_v$ in cache. \medskip
    \STATE Collect  measurements \medskip
    \STATE At the end of the period, compute estimate $z_{v}$ of $\partial_{y_v} L(Y^{(k)})$ through \eqref{estimation}.
    \STATE Compute new state  $y_{v}^{(k+1)}$ through \eqref{adapt}.
  \ENDFOR
  \end{algorithmic}
\end{small}
\end{algorithm}

\subsubsection{Algorithm Overview}
 We begin by giving an overview of our distributed, adaptive algorithm. We partition time into periods of equal length $T>0$, during which each node  $v\in V$ collects measurements from messages routed through it.  Each node keeps track of \emph{its own marginals}  $y_v\in[0,1]^{|\catalog|}$: intuitively, as in \eqref{marginals}, each $y_{vi}$ captures the probability that node $v\in V$ stores item $i\in\catalog$.  We refer to $y_v$ as the \emph{state} at node $v$; these values, as well as the cache contents of a node, remain constant  during a measurement period. 
When the period ends, each node (a) adapts its state vector $y_v$, and (b) reshuffles the contents of its cache, in a manner we describe below. 

\sloppy
In short, at any point in time, the  (global) allocation $X\in\feasibledomain_1$ is sampled from a joint distribution $\mu$ that has a \emph{product form}; for every $v$, there exist  appropriate probability distributions $\mu_v$, $v\in V$, such that:
\begin{align}\mu(X) = \textstyle\prod_{v\in V}\mu_v(x_{v1},\ldots, x_{v|\catalog|}).\label{productform}\end{align} 
Moreover, each marginal probability $\prob_{\mu}[x_{vi}=1]$ (i.e., the probability that node $v$ stores $i$) is determined as a ``smoothened'' version of the current state variable $y_{vi}$.

\fussy
\noindent\textbf{State Adaptation.} A node  $v\in V$ uses local measurements collected from messages it receives during a period to produce a random vector $z_{v}\in \reals_+^{|\catalog|}$ that is an unbiased estimator of a subgradient of $L$ w.r.t.~to $y_v$.  That is, if $Y^{(k)}\in\reals^{|V|\times |\catalog|}$ is the (global) matrix of marginals at the $k$-th measurement period, $z_{vi}=z_{vi}(Y^{(k)})$ is a random variable satisfying:
 \begin{align}\label{estimateprop}
\expect\left[z_{v}(Y^{(k)})\right] \in \partial_{y_{v}}L(Y^{(k)})\end{align}
where $\partial_{y_{v}} L(Y)$ is the set of subgradients of $L$ w.r.t $y_{v}$.  We specify how to produce such estimates in a distributed fashion  below, in Section~\ref{sec:distributedsub}. 
 Having these estimates, each node adapts its state as follows: at the conclusion of the $k$-th period, the new state is computed as
 \begin{align}y_v^{(k+1)} \leftarrow \mathcal{P}_{\feasibledomain_2^v} \left( y_v^{(k)} + \gamma_k\cdot z_v(Y^{(k)}) \right),\label{adapt}\end{align}
where $\gamma_k>0$ is a gain factor and $\mathcal{P}_{\feasibledomain^v_2}$ is the projection to $v$'s set of relaxed constraints: 
$$\feasibledomain_2^v = \{y_v\in [0,1]^{|\catalog|} : \sum_{i\in \catalog}y_{vi}=c_v, y_{vi}=1, \text{for }i\text{ s.t~}v\in \source_i\}.$$

\noindent\textbf{State Smoothening.}
Upon performing the state adaptation \eqref{adapt}, each node $v\in V$ computes the following
``sliding average'' of its current and past states:
 \begin{align}\label{slide}\bar{y}_v^{(k)} = \textstyle \sum_{\ell = \lfloor\frac{k}{2} \rfloor}^{k} \gamma_\ell y_v^{(\ell)} /\left[\sum_{\ell=\lfloor \frac{k}{2}\rfloor}^{k}\gamma_{\ell}\right]   .\end{align}
 This ``state smoothening''  is necessary precisely because of the non-differentiability of $L$ \cite{nemirovski2005efficient}. Note that  $\bar{y}_v^{(k)} \in \feasibledomain_2^v$, as a convex combination of points in $\feasibledomain_2^v$.

\noindent\textbf{Cache reshuffling.} Finally, given  $\bar{y}_v$, each node $v\in V$ reshuffles its contents, placing items in its cache \emph{independently of all other nodes}: that is, node $v$ selects a random allocation $x_v^{(k)}\in \{0,1\}^{|\catalog|}$ sampled independently of any other node in $V$, so that the joint distribution satisfies \eqref{productform}. 

In particular, $x_v^{(k)}$ is sampled from a distribution $\mu_v$ that has the following two properties:
\begin{packedenumerate}
\item $\mu_v$ is a distribution over \emph{feasible} allocations, satisfying $v$'s capacity and source constraints, i.e., $\mu_v$'s support is:
$$\!\!\!\!\!\!\feasibledomain_1^v = \{x_v\!\in\! \{0,1\}^{|\catalog|} : \sum_{j\in \catalog}\!x_{vj}\!=\!c_v, x_{vi}\!=\!1, \text{for }i\text{ s.t~}v\in \source_i\}.$$
\item $\mu_v$ is \emph{consistent} with the marginals $\bar{y}_v^{(k)}\in\feasibledomain^v_2 $, i.e., 
\begin{align}\expect_{\mu_v}[x_{vi}^{(k)}]=\bar{y}_{vi}^{(k)},\text{ for }i\in\catalog.\label{sample}\end{align}
\end{packedenumerate}
It is not obvious that a $\mu_v$ that satisfies these properties exists. We show below, in Section~\ref{sec:derandimize}, that such a $\mu_v$  exists,  it has $O(|\catalog|)$  support, and can be computed
 in $O(c_v|\catalog|\log|\catalog|)$ time. As a result, having $\bar{y}_v$, each node $v$ can sample a feasible allocation $x_v$ from   $\mu_v$ in $O(c_v|\catalog|\log|\catalog|)$ time. 

The complete projected gradient ascent algorithm is summarized in Algorithm~\ref{alg:ascent}.
A node may need to retrieve new items, not presently in its cache,  to implement the sampled allocation $x_v^{(k)}$.  This incurs additional routing costs but, if $T$ is large, this traffic is small compared to regular response message traffic. Before we state its convergence properties, we present the two missing pieces, namely,  the subradient estimation and  randomized rounding we use.

\subsubsection{Distributed Sub-Gradient Estimation}\label{sec:distributedsub} We now describe how to compute the estimates $z_v$ of the subgradients $\partial_{y_v}L(Y)$ in a measurement period, dropping the superscript $\cdot^{(k)}$ for brevity. These estimates are computed in a distributed fashion at each node, using only information available from messages traversing the node. This computation requires additional ``control'' messages to be exchanged between nodes, beyond the usual request and response traffic.

The estimation proceeds as follows. Given a path $p$ and a $v\in p$, denote by  
$k_p(v)$ is the position of $v$ in $p$; i.e., $k_p(v)$ equals the $k\in \{1,\ldots,|p|\}$ such that $p_k=v$. Then:
\begin{packedenumerate}
\item  Every time a node generates a new request $(i,p)$, it also generates an additional control message to be propagated over $p$, in parallel to the request. This message is propagated until a node $u\in p$ s.t.~$\textstyle\sum_{\ell=1}^{k_{p}(u)} y_{p_{\ell}i}> 1$ is found, or the end of the path is reached. This can be detected by summing the state variables $y_{vi}$ as the control message traverses nodes $v\in p$ up to $u$.
\item Upon reaching either such a  node or the end of the path, the control message is sent down in the reverse direction. Every time it traverses an edge in this reverse direction, it adds the weight of this edge into a weight counter. 
\item Every node on the reverse path ``sniffs'' the weight counter field of the control message. Hence, every node visited learns the sum of weights of all edges connecting it to visited nodes further in the path; i.e., visited node $v\in p$ learns the quantity: $$ t_{vi} = \textstyle\sum_{k'=k_v(p)}^{|p|-1} w_{p_{k'+1}p_{k'}} \id_{\sum_{\ell=1}^{k'} y_{p_{\ell}i}\leq 1}.$$ 
\item Let $\mathcal{T}_{vi}$ be the set of quantities collected in this way at node $v$ regarding item $i\in \mathcal{C}$ during a measurement period of duration $T$. At the end of the measurement period, each node $v\in V$ produces the following estimates: \begin{align}z_{vi}=\frac{1}{T} \textstyle\sum_{t\in \mathcal{T}_{vi} }t,\quad i\in\catalog.\label{estimation}\end{align} Note that, in practice, this needs to be computed only for $i\in \catalog$ for which $v$ has received a control message. 
\end{packedenumerate}
Note that the control messages in the above construction are ``free'' under our model, as they do not carry an item. Moreover, they can be piggy-backed on/merged with request/response messages, wherever the corresponding traversed sub-paths of $p$ overlap.
It is easy to show that the above estimate is an unbiased estimator of the subgradient: \begin{lemma}\label{subgradientlemma}
For  $z_{v}=[z_{vi}]_{i\in \catalog}\in \reals_+^{|\catalog|}$ the vector constructed through coordinates \eqref{estimation},
$$\expect[z_v(Y)] \in \partial_{y_v}L(Y)\text{ and }\expect[\|z_v\|_2^2]<W^2|V|^2|\catalog|(\Lambda^2+\frac{\Lambda}{T}),$$
where 
$W=\displaystyle\max_{(i,j)\in E}w_{ij}$ and  $\Lambda =\displaystyle\max_{v\in V,i\in\catalog}\!\!\sum_{(i,p)\in \requests: v\in p } \!\!\!\!\lambda_{(i,p)}.$
\end{lemma}
\begin{proof}
 First, let:
\begin{subequations}\label{partiallimits}
\begin{align}
\overline{\partial_{y_{vi}}L}(Y) &= \!\!\! \!\!\!\!\sum_{(i,p)\in \requests:v\in p}\!\!\!\!\! \!\!\lambda_{(i,p)}\!\!\!\!\sum_{k'=k_p(v)}^{|p|-1} \!\!\!w_{p_{k'+1}p_{k'}} \id_{\sum_{\ell=1}^{k'} y_{p_{\ell}i}\leq 1}\\
\underline{\partial_{y_{vi}}L}(Y) & = \!\!\!\!\!\!\sum_{(i,p)\in \requests:v\in p}\!\!\!\!\!\!\! \lambda_{(i,p)}\!\!\!\!\sum_{k'=k_p(v)}^{|p|-1} \!\!\!w_{p_{k'+1}p_{k'}} \id_{\sum_{\ell=1}^{k'} y_{p_{\ell}i}< 1 }
\label{lower}\end{align}
\end{subequations}
where 
$k_p(v)$ is the position of $v$ in $p$, i.e., it is the $k\in \{1,\ldots,|p|\}$ such that $p_k=v$.

A vector $z\in \reals^{|\catalog|}$  belongs to the subgradient set $\partial_{y_v}L(Y)$ if and only if
$z_i \in [ \underline{\partial_{y_{vi}}L}(Y) ,\overline{\partial_{y_{vi}}L}(Y) ]. $
If $L$ is differentiable at $Y$ w.r.t $y_{vi}$, the two limits coincide and are equal to $\tfrac{\partial L}{\partial y_{vi}}.$
It immediately follows from the fact that requests are Poisson that $\expect[z_{vi}(Y)] = \overline{\partial_{y_{vi}}L}(Y)$, so indeed $\expect[z_{v}(Y)] \in \partial_{y_v}L(Y)$.
To prove the bound on the second moment, note that
$\expect[z_{vi}^2]=\frac{1}{T^2}\expect[(\sum_{t\in \mathcal{T}_{vi}}t)^2] 
\leq \frac{W^2|V|^2}{T^2} \expect\Big[|\mathcal{T}_{vi}|^2\Big]$
as $t\leq W|V|$. On the other hand, $|\mathcal{T}_{vi}|$ is Poisson distributed with expectation $\sum_{(i,p)\in \requests:v\in p}\lambda_{(i,p)} T$, and the lemma follows. \hspace*{\stretch{1}}\qed
\end{proof}

\subsubsection{Distributed Randomized Rounding}\label{sec:derandimize}
We now turn our attention to the distributed, randomized rounding scheme executed each node $v\in V$. To produce a $\mu_v$ over $\feasibledomain_1^v$ that satisfies \eqref{sample}, note that it suffices to consider $\mu_v$ such that $\expect_{\mu_v}[x_v]=\bar{y}_{v}$, defined  over the set: \begin{align}
\bar{\feasibledomain}_1^v =\{x_v\in  \{0,1\}^{|\catalog|}:\textstyle\sum_{i\in \catalog} x_{vi} = c_v\}.
\label{localconstr}
\end{align}
 That is, subject to attaining correct marginals, one can ignore the source constraints:
to see this, note that if $v\in S_i$,  $\bar{y}_{vi}=1$ for any $Y\in \feasibledomain_2$.  Hence, \eqref{sample}  ensures that $v$ stores item $i$ w.p.~1. We thus focus on constructing a distribution $\mu_v$ over $\bar{\feasibledomain}_1^v$, under a given set of marginals $\bar{y}_v$.
Note that a ``na\"ive'' construction in which  $x_{vi}$, $v\in V$, $i\in\catalog$,  are independent Bernoulli variables with parameters $\bar{y}_{vi}$ indeed satisfies \eqref{sample}, but \emph{does not} yield vectors $x_v\in\bar{\feasibledomain}_1^v$: indeed, such vectors only satisfy the capacity constraint in expectation, and may contain fewer or more items than $c_v$.

\begin{figure}[t]
\hspace*{\stretch{1}}\includegraphics[width=0.9\columnwidth]{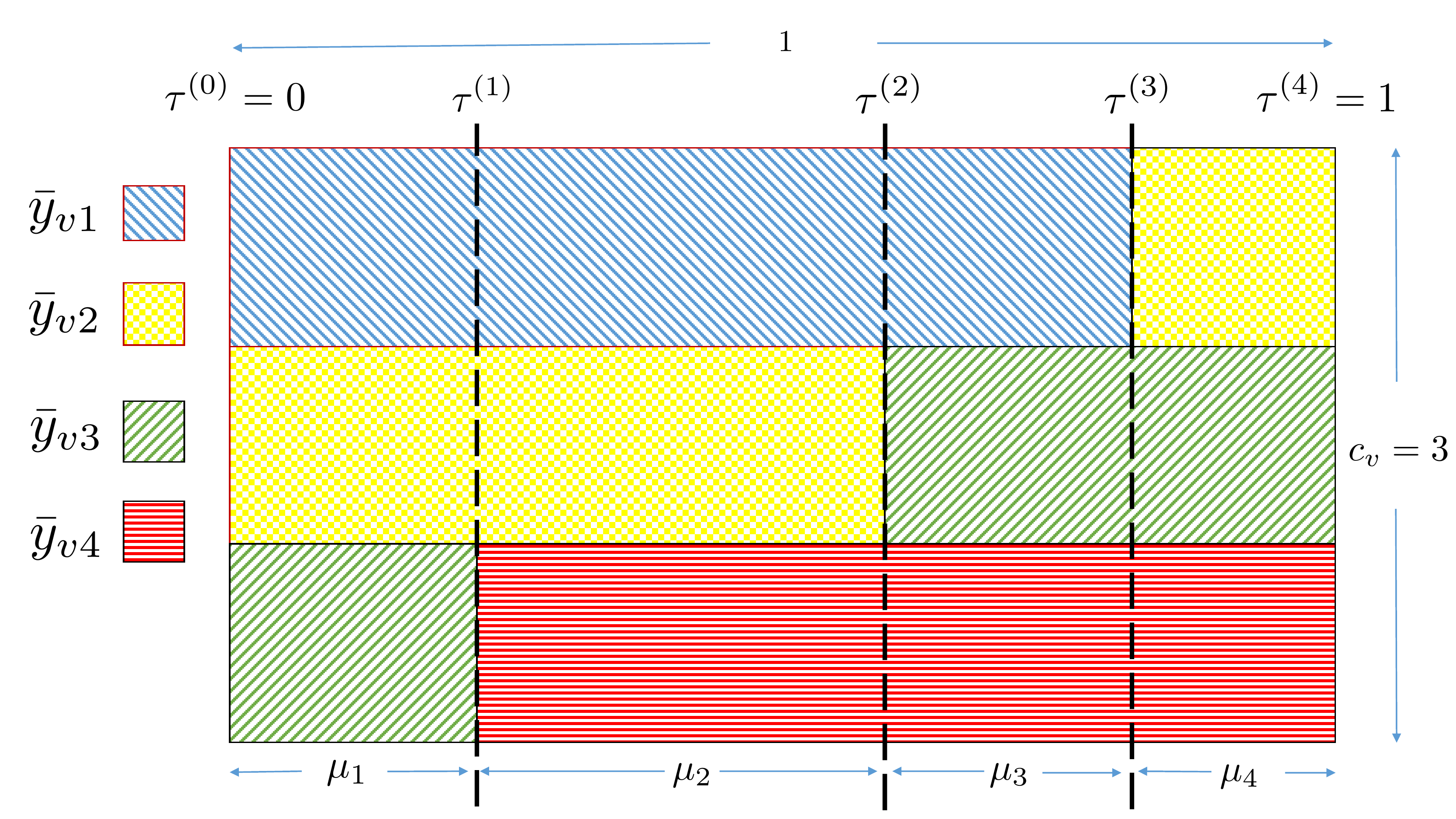}\hspace*{\stretch{1}}
\caption{An allocation that satisfies $\expect_{\mu}[x_{vi}]=\bar{y}_{vi}$, when $\sum_{i\in \catalog}\bar{y}_{vi}=c_v$. After placing the 4 rectangles in a $1\times 3$ grid, assigning probabilities $\mu_1$, $\mu_2$, $\mu_3$, $\mu_4$ to each of the tuples $\{1,2,3\}$, $\{1,2,4\}$, $\{1,3,4\}$, $\{2,3,4\}$, respectively, yields the desired marginals. }\label{fig:allocation}
\end{figure}

\sloppy
Before we formally present our algorithm we first give some intuition behind it, also illustrated in Figure~\ref{fig:allocation}. Let $c_v=3$, $\catalog=\{1,2,3,4\}$, and consider a $\bar{y}_v\in \feasibledomain_2^v.$ To construct an allocation with the desired marginal distribution, consider  a rectangle box  of area $c_v\times 1$. For each $i\in \catalog$, place a rectangle of length $\bar{y}_{vi}$ and height 1 inside the box, starting from the top left corner. If a rectangle does not fit in a row, cut it, and place the remainder in the row immediately below, starting again from the left. As $\sum_{i=1}^{|\catalog|}\bar{y}_{vi}=c_v$, this space-filling method completely fills (i.e., tessellates) the $c_v\times 1$ box.

\fussy
Consider now, for each row, all fractional values $\tau\in[0,1]$ at which two horizontal rectangles meet. We call these values the \emph{cutting points}. Notice that there can be at most $|\catalog|-1$ such points. Then, partition the $c_v\times 1$ box vertically, splitting it at these cutting points. This results in at most $|\catalog|$ vertical partitions (also rectangles), with $c_v$ rows each.
Note that each one of these vertical partitions correspond to tuples comprising $c_v$ distinct items of $\catalog$. Each row of a vertical partition must contain some portion of a horizontal rectangle, as the latter tessellate the entire box. Moreover,  no vertical partition can contain the same horizontal rectangle in two rows or more (i.e., a horizontal rectangle cannot ``overlap'' with itself), because $\bar{y}_{vi}\leq 1$, for all $\bar{y}_{vi}\in\catalog$. 
The desired probability distribution $\mu_v$ can then be constructed by setting  (a) its support to be the $c_v$-tuples defined by each vertical partition, and (b) the probability of each $c$-tuple  to be the length of the partition (i.e., the difference of the two consecutive $\tau$ cutting points that define it). The marginal probability of an item will then be exactly the length of its horizontal rectangle, i.e., $y_{vi}$, as desired.
 
The above process is described formally in Algorithm~\ref{alg:allocation}. The following lemma establishes its correctness:
\begin{lemma} \label{placementcorrectness}
 Alg.~\ref{alg:allocation} produces a $\mu_v$ over $\bar{\feasibledomain}_1^v$ s.t.~\eqref{sample} holds.
\end{lemma}
\techrep{A proof can be found in Appendix~\ref{appendix:corr}.}{A proof can be found in \cite{arxiv}.} Note that, contrary to the ``na\"ive'' Bernoulli solution, the resulting  variables $x_{vi},x_{vj}$, where $i\neq j$,  \emph{may not be independent} (even though allocations are independent across caches).
 The algorithm's complexity is  $O(c_v|\catalog|\log|\catalog|)$, as the sort in line \ref{sort} can be implemented in $O(c_v|\catalog|\log|\catalog|)$ time, while  a match in \ref{find} can be found in $O(\log|\catalog|)$ time if intervals are stored in a binary search tree. Moreover, the support of $\mu_v$ has size at most $|\catalog|$, so representing this distribution requires $O(c_v|\catalog|)$ space.

\begin{algorithm}[t]
\begin{small}
  \caption{\textsc{Placement Algorithm}}\label{alg:allocation}
    \begin{algorithmic}[1]
  \STATE \textbf{Input:} capacity $c_v$, marginals $\bar{y}_v\in \reals^{|\catalog|}$ s.t. $\bar{y}_v\geq\mathbf{0}$, $\sum_{i=1}^{|\catalog|}\bar{y}_{vi}=c_v$
  \STATE \textbf{Output:} prob.~distr.~$\mu_v$ over $\{x\in\{0,1\}^{|\catalog|}:\sum_{i=1}^{|\catalog|}=c_v \}$ s.t.~$\expect_{\mu}[x_i] = \bar{y}_{vi},$ for all $i\in \catalog$. 
   \newcommand{\sofar}{\ensuremath{\mathtt{sum}}}
    \STATE $\sofar\leftarrow 0$
    \FORALL{$i\in \catalog$}
	\STATE	$s_i \leftarrow \sofar$ ;  $t_i \leftarrow \sofar +\bar{y}_{vi}$; $\tau_i \leftarrow t_i-\lfloor t_i\rfloor$ \label{quantities}
	\STATE $\sofar\leftarrow t_i$
    \ENDFOR
    \STATE Sort all $\tau_i$ in increasing order, remove duplicates, and append $1$ to the end of the sequence.\label{sort}
    \STATE Let $0=\tau^{(0)}<\tau^{(1)}<\ldots<\tau^{(K)}=1$ be the resulting sequence.\label{sorted}
    \FORALL{ $k \in \{0,\ldots,K-1\} $}
    \STATE Create new vector $x\in \{0,1\}^{|\catalog|}$; set $x\leftarrow\mathbf{0}$.
    \FORALL{$\ell\in\{0,\ldots,c-1\}$}
	\STATE Find $i\in \catalog$ such that $(\ell+\tau^{(k)},\ell+\tau^{(k+1)})\subset [s_i,t_i]$. \label{find}
	\STATE Set $x_i\leftarrow 1$
    \ENDFOR
    \STATE Set $\mu_v(x) = \tau^{(k+1)}-\tau^{(k)}$
    \ENDFOR
  \RETURN $\mu_v$
  \end{algorithmic}
\end{small}
\end{algorithm}

\subsubsection{Convergence}  We now establish the convergence  of the smoothened marginals of projected gradient ascent to a global minimizer of $L$:
\begin{thm} \label{convergencethm}Let $\bar{Y}^{(k)}\in \feasibledomain_2$ be the smoothened marginals at the $k$-th period of Algorithm~\ref{alg:ascent}. Then,
$$\varepsilon_k\equiv \expect[\max_{Y\in \feasibledomain_2} L(Y)-L(\bar{Y}^{(k)})] \leq  \frac{D^2 + M^2 \sum_{\ell=\lfloor k/2\rfloor}^{k}\gamma_\ell^2 }{\sum_{\ell=\lfloor k/2\rfloor}^{k}\gamma_\ell} ,$$
where $D = \sqrt{2|V|\max_{v\in V}c_v}$, $M=W|V| \Lambda\sqrt{|V||\catalog|(1+\frac{1}{\Lambda T})}$. In particular, for $\gamma_k = \tfrac{D}{M\sqrt{k}}$, we have $\varepsilon_{k}\leq O(1) \frac{MD}{\sqrt{k}},$ where $O(1)$ is an absolute constant.
\end{thm}
\begin{proof}
Under dynamics \eqref{adapt} and \eqref{slide}, from Theorem 14.1.1, page 215 of Nemirofski~\cite{nemirovski2005efficient}, we have that
$$ \expect[\max_{Y\in \feasibledomain_2}L(Y) - L(\bar{Y}^{(k)})] \leq  \frac{D^2 + M^2 \sum_{\ell=\lfloor k/2\rfloor}^{k}\gamma_\ell^2 }{\sum_{\ell=\lfloor k/2\rfloor}^{k}\gamma_\ell} $$
where 
$D\equiv\max_{x,y\in\feasibledomain_2}\|x-y\|_2 = \sqrt{\max_v 2|V| c_v}$ is  the diameter of $\feasibledomain_2$, and
$$M \equiv\sup_{Y\in\feasibledomain_2} \sqrt{ \sum_{v\in V} \expect[\|z_v(Y)\|_2^2] } \leq  W|V|\sqrt{|V||\catalog|(\Lambda^2+\frac{\Lambda}{T})},$$ where the last equality follows from Lemma~\ref{subgradientlemma}.\hspace*{\stretch{1}}\qed
\end{proof}
The  $O(1) \frac{MD}{\sqrt{k}}$ upper bound presumes knowledge of $D$ and $M$ when setting the gains $\gamma_k$, $k\geq 1$. Nonetheless, even when these are not apriori known, taking 
$\gamma_{k} =1/\sqrt{k}$
suffices to ensure the  algorithm converges with rate $1/\sqrt{k}$, up to (larger) constants, that depend on $D$ and $M$. 
Moreover, the relationship between $M$ and $T$ captures the tradeoff induced by $T$: larger $T$s give more accurate estimates of the subgradients, reducing the overall number of steps till convergence, but increase the length of each individual period.
Finally,  Thms.~\ref{approximation} and~\ref{convergencethm} imply that the asymptotic expected caching gain under Algorithm~\ref{alg:ascent} is within a constant factor from the optimal: \begin{thm}\label{maincor}
Let $X^{(k)}\in \feasibledomain_1$ be the allocation at the $k$-th period of Algorithm~\ref{alg:ascent}. Then, if $\gamma_k =\Theta(1/\sqrt{k})$, 
$$\lim_{k\to \infty} \expect[F(X^{(k)})] \geq \big(1 -\frac{1}{e}\big)\max_{X\in \feasibledomain_1} F(X). $$
\end{thm}
\begin{proof} \sloppy
Given $\bar{Y}^{(k)}$,  by Lemma~\ref{placementcorrectness},  $X^{(k)}$ is sampled from a distribution $\mu$ over $\feasibledomain_1$ that has product form \eqref{productform}.  This product form implies that, conditioned on  $\bar{Y}^{(k)}$, Eq.~\eqref{equality} holds;  thus, $\expect[F(X^{(k)})\mid \bar{Y}^{(k)}] = \expect[F(\bar{Y}^{(k)})]$, so
$\lim_{k\to\infty} \expect[F(X^{(k)})] =\lim_{k\to \infty }\expect[F(\bar{Y}^{(k)})].$
From Thm.~\ref{convergencethm},
$\lim_{k\to\infty}\expect[L(\bar{Y}^{(k)})] =\max_{Y\in \feasibledomain_2} L(Y).$
This implies that, for $\nu^{(k)}$ the distribution of $\bar{Y}^{(k)}$, and $\Omega$ the set of $Y\in \feasibledomain_2$ that are maximizers of $L$, 
 $\lim_{k\to\infty} \nu^{(k)}(\feasibledomain_2\setminus \Omega )=0.$
From Theorem~\ref{approximation},   $F(Y)\geq (1-1/e)\max_{X\in  \feasibledomain_1} F(X)$ for any $Y\in \Omega$. The theorem therefore follows from the above observations, and the fact that $F$ is bounded in $\feasibledomain_2\setminus \Omega$.
\end{proof}

\fussy
We state these results under stationary demands but, in practice, we would prefer that caches adapt to demand fluctuations. To achieve this, one would fix $\gamma$ to a constant positive value, ensuring that  Algorithm~\ref{alg:ascent} tracks  demand changes. Though convergence to a minimizer is not guaranteed in this case, the algorithm is nonetheless guaranteed to reach states concentrated around an optimal allocation (see, e.g., Chapter 8 of Kushner \& Yin \cite{kushner2003stochastic}).

\subsection{Greedy Path Replication}\label{sec:greedy}
Algorithm~\ref{alg:ascent} has certain drawbacks. First, to implement an allocation  at the end of a measurement period, nodes may need to retrieve new items, which itself incurs additional traffic costs. Second, there is a timescale separation between how often requests arrive, and when adaptations happen; an algorithm that adapts at the request timescale can potentially converge faster. Third, caches are synchronized, and avoiding such coordination is preferable.  Finally, beyond request and response messages, the exchange of additional control messages are required to implement it.

In this section, we propose a greedy eviction policy, to be used with the path replication algorithm, that has none of the above drawbacks. This algorithm does not require any control traffic beyond the traffic generated by message exchanges.  It is asynchronous, and its adaptations happen at the same timescale as requests.
 Each node makes caching decisions \emph{only} when it receives a response message carrying an item: that is, a node decides whether to store an item exactly when it passes through it, and does not introduce additional traffic to retrieve it. Finally, the eviction heuristic is very simple (though harder to analyze than Algorithm~\ref{alg:ascent}).

\begin{algorithm}[t]
\begin{small}
  \caption{\textsc{Greedy Path Replication}}\label{alg:greedy}
    \begin{algorithmic}[1]
   \STATE Execute the following at each $v\in V$:\medskip
   \STATE Initialize $z_v=\mathbf{0}$.
   \WHILE{ (\texttt{true})}    \STATE Wait for new response message.
    \STATE Upon receipt of new message, extract counter $t_{vi}$ and $i$.
    \STATE Update $z_v$ through \eqref{EWMA}.
    \STATE Sort $z_{vj}$, $j\in \catalog$, in decreasing order.
    \IF {$x_{vi}=0$ and $i$ is in top $c_v'$ items}
    \STATE Set $x_{vi}\leftarrow 1$; evict $c_v'+1$-th item.
    \ENDIF
  \ENDWHILE
  \end{algorithmic}
\end{small}
\end{algorithm}

\subsubsection{Algorithm Overview}
Node states are determined by cache contents, i.e., vectors $x_v\in \feasibledomain_1^v$, $v\in V$ and, at all times, $v$ stores all $i$ s.t. $v\in S_i$. The algorithm then proceeds as follows: 
 \begin{packedenumerate}
\item As usual, for each $(i,p)\in\requests$, request messages are propagated until they reach a node $u$ caching the requested item $i\in \catalog$, i.e., with $x_{ui}=1$. Once reaching such a node, a response message carrying the item is generated, and backpropagated over $p$.
\item The response message for a request $(i,p)$ contains a weight counter that is initialized to zero by $u$. Whenever the response traverses an edge in the reverse path, the edge weight is added to the counter. The counter is ``sniffed'' by every node in $p$ that receives the response. Hence, every node $v$ in the path $p$ that is visited by a response learns the quantity:
\begin{align}
 t_{vi} = \textstyle\sum_{k'=k_v(p)}^{|p|-1} w_{p_{k'+1}p_{k'}} \id_{\sum_{\ell=1}^{k'} x_{p_{\ell}i}<1}, \label{quantity}
\end{align}
where, as before, $k_v(p)$ is the position of $v$ in path $p$.
\item  For each item $i$, each node $v$ maintains again an estimate $z_{v} \in \reals_+^{|\catalog|} $ of a subgradient in $\partial_{x_{vi}} L(X)$. This estimate is maintained through an \emph{exponentially weighted moving average} (EWMA) of the quantities $t_{vi}$ collected above. These are adapted each time $v$ receives a response message. If $v$ receives a response message for $i$ at time $t$, then it adapts its estimates as follows:  for all $j\in\mathcal{C}$,
\begin{align}\label{EWMA}
z_{vj}(t)= z_{vj}(t') \cdot e^{-\beta (t-t')}+\beta \cdot t_{vi}\cdot\id_{i=j}, 
\end{align} 
where $\beta>0$ is the EWMA gain, and $t'<t$ be the last time node $v$ it received a response message prior to $t$. 
 Put differently, estimates $z_{vj}$ decay exponentially between responses, while only $z_{vi}$, corresponding to the requested item $i$, contains an additional increment.  \item After receiving a response and adapting $z_v$, the node (a) sorts $z_{vi}$, $i\in C$, in a decreasing order, and (b) stores the top $c_v'$ items, where
 $c_v' = c_v -|\{i: v\in S_i\}|$
is $v$'s capacity excluding  permanent items. \end{packedenumerate}
The above steps are summarized in Algorithm~\ref{alg:greedy}. Note that, upon the arrival of a response carrying item $i$, there are only two possible outcomes after the new $z_v$ values are sorted: either (a) the cache contents remain unaltered, or (b)  item $i$, which was previously not in the cache, is now placed in the cache, and another item is evicted. These are the only possibilities because, under \eqref{EWMA}, all items $j\neq i$ \emph{preserve their relative order}: the only item whose relative position may change is $i$. As $i$ is piggy-backed in the response,  no additional traffic is needed to acquire it.

\subsubsection{Formal Properties} Though simpler to describe and implement, Algorithm~\ref{alg:greedy} harder to analyze than Algorithm~\ref{alg:ascent}. Nonetheless, some intuition on its performance can be made by looking into its fluid dynamics. In particular, let $X(t)=[x_{vi}(t)]_{v\in V,i\in \catalog}$ and $Z(t)=[z_{vi}(t)]_{v\in V,i\in \catalog}$ be the allocation and subgradient estimation matrices at time $t\geq 0$.  Ignoring stochasticity, the ``fluid'' trajectories of these matrices are described by the following ODE:
\begin{subequations}\label{ODE}
\begin{align}
X(t) &\in \argmax_{X\in \feasibledomain_1} \langle X, Z(t) \rangle \label{greedystep}\\
\frac{dZ(t)}{dt}& = \beta\big(\underline{\partial L}(X(t)) - Z(t)\big)
\end{align}
\end{subequations}
 where $\langle A,B\rangle = \trace(AB^\top) = \sum_{v,i}A_{vi}B_{vi}$ is the inner product between  two matrices, and $\underline{\partial L}\in \partial L $ is a subgradient of $L$ at $X$.  \begin{proof}
Using the ``baby Bernoulli'' approximation of a Poisson process, and the fact that $e^{x}=1+x+o(x)$ the EWMA adaptations have the following form: for small enough $\delta>0$,
$$z_{vi}(t+\delta) = (1-\beta\delta) z_{v_i}(t) + \beta \delta \phi_{vi} +o(\delta)$$
where $\expect[\phi_{vi}]=\underline{\partial_{x_{vi}}}L(X)$,  and $\underline{\partial_{x_{vi}} L}$ is given by \eqref{lower}; note that  the matrix $\underline{\partial L}$ is indeed a subgradient. The fluid dynamics \eqref{ODE} then follow by taking $\delta$ to go to zero, and replacing the $\phi_{v}$'s by their expectation. 
\end{proof}
The dynamics \eqref{ODE} are similar (but not identical) to the  ``continuous greedy'' algorithm for submodular maximization \cite{vondrak2008optimal} and the Frank-Wolfe algorithm~\cite{clarkson2010coresets}.

ODE \eqref{ODE} implies that EWMA $Z(t)$ indeed ``tracks'' a subgradient of $L$ at $X$. On the other hand, the allocation selected by sorting---or, equivalently, by \eqref{greedystep}---identifies the most valuable items at each cache w.r.t.~the present estimate of the subgradient.
Hence, even if $z_{v}$ is an inaccurate estimate of the subgradient at $v$, the algorithm  treats it as a correct estimate and places the ``most valuable'' items in its cache. This is why we refer to this algorithm as ``greedy''.  Note that \eqref{greedystep} also implies that $$X(t) \in  \argmax_{Y\in \feasibledomain_2} \langle Y, Z(t) \rangle.  $$
This is because, subject to the capacity and source constraints, $\langle \cdot ,Z\rangle$ is maximized by taking any set of top $c_v'$ items, so an integral solution indeed always exists. The following lemma states that fixed points of the ODE \eqref{ODE}, provided they exist, must be at maximizers of $L$.
\begin{lemma}\label{fixed:point}
Let $X^*\in \feasibledomain_2 $ and $Z^*\in \reals^{|V|\times |\catalog|}$ be such that
$X^*\in \textstyle \argmax_{X\in\feasibledomain_2} \langle X,Z^*\rangle$ and 
$Z^* \in \partial L(X^*).$
Then, $X^*\in\argmax_{X\in \feasibledomain_2} L(X)$.
\end{lemma}
The lemma holds by the concavity of $L$, and is stated as  Theorem 27.4 of Rockafellar~\cite{rockafellar}, so we omit its proof. Though the conditions stated in the lemma are both necessary and sufficient in our case, the lemma does not imply that a \emph{integral} solution (i.e., one in which $X^*\in \feasibledomain_1$) need exist. In practice, the algorithm may converge to a chain-recurrent set of integral solutions. Though we do not study the optimality properties of this set formally,  our numerical evaluations in Section~\ref{sec:numerical} show that this greedy heuristic has excellent performance in practice, very close to the one attained by the maximizer of $L$.

\section{Offline Problem Equivalence}\label{sec:equivalence}

Each iteration of the projected gradient ascent algorithm of Section~\ref{sec:adaptive} constructs probabilistic allocations that are (a) feasible, and (b) independent across nodes. This motivates us to study the following probabilistic relaxations of \CG, beyond the ``independent Bernoulli'' relaxation \eqref{indep} we discussed in Section~\ref{sec:pipage}. First, consider the variant:
\begin{subequations}\label{independentcaches}
\begin{align}
\text{Max.:}& \quad \expect_\mu[F(X)] = \textstyle\sum_{X\in \feasibledomain_1}\mu(X)F(X)\\
\text{subj.~to:}& \quad \mu\text{ is a pr.~distr.~over }\feasibledomain_1\text{ satisfying }\eqref{productform}.
\end{align}
\end{subequations}
I.e., we seek random cache allocations sampled from a joint distribution $\mu$ having product form \eqref{productform}. In addition, consider the following (more general) variant of \CG:
\begin{subequations}\label{probabilistic}
\begin{align}
\text{Maximize:}& \quad \expect_\mu[F(X)] = \textstyle\sum_{X\in \feasibledomain_1}\mu(X)F(X)\\
\text{subj.~to:}& \quad \mu\text{ is a pr.~distr.~over }\feasibledomain_1.
\end{align}
\end{subequations}
Our results and, in particular, Lemma~\ref{placementcorrectness}, have the following surprising implication: all three relaxations \eqref{indep}, \eqref{independentcaches}, and \eqref{probabilistic} \emph{are in fact equivalent to} \CG.  
\begin{thm}\label{equivalence}Let $X^*$,   $Y^{*}$, $\mu^*$, and $\mu^{**}$ be  optimal solutions to \eqref{deterministic},\eqref{indep}, \eqref{independentcaches}, and \eqref{probabilistic}, respectively. Then,
$$ F(X^*)  = F(Y^*)=\expect_{\mu^*}[F(X)]=\expect_{\mu^{**}}[F(X)].$$
\end{thm}
\begin{proof}
We  establish the following inequalities:
$$\expect_{\mu^{**}}[F(X)] \leq F(X^*) \leq F(Y^*) \leq \expect_{\mu^*}[F(x)] \leq \expect_{\mu^{**}}[F(X)]$$
To see that $\expect_{\mu^{**}}[F(X)] \leq F(X^*)$, let $D =\supp(\mu^{**})\subseteq \feasibledomain_1$ be the support of $\mu^*$. Let 
$X'\in \argmax_{X\in D} F(X)$ be an allocation maximizing  $F$ over $D$ (as $D$ is finite, this exists). Then, by construction, $\expect_{\mu^{**}}[F(X)]\leq F(X')\leq F(X^*)$, as $X'\in \feasibledomain_1$. $F(X^*)\leq F(Y^*)$ by \eqref{trivial}, as \eqref{indep} is a relaxation of \eqref{deterministic}.
To see  that $F(Y^*) \leq \expect_{\mu^*}[F(X)]$, note that, by Lemma~\ref{placementcorrectness}, since $Y^*\in \feasibledomain_2$, there exists a measure $\mu'$ that has a product form and whose marginals are $Y^*$. Since $\mu'$ has a product form, it satisfies \eqref{equality}, and $F(Y^*)=\expect_{\mu'}[F(X)]\leq \expect_{\mu^*}[F(X)]$.  Finally, $\expect_{\mu^*}[F(X)]\leq \expect_{\mu^{**}}[F(X)]$, as the former is the expected cost under a restricted class of distributions $\mu$, namely, ones that have the product form \eqref{productform}.\hspace{\stretch{1}}\qed
\end{proof}
 Theorem~\ref{equivalence} is specific to \CG: e.g., Ineq.~\eqref{trivial} can be strict in other problems solvable through the pipage rounding method. The theorem has some non-obvious, interesting implications. First, equivalence of \eqref{independentcaches} to \eqref{indep} implies that satisfying capacity constraints in expectation, rather than exactly, \emph{does not} improve the caching gain. Similarly, the equivalence of  \eqref{independentcaches} to \eqref{probabilistic} implies that considering only distributions that describe \emph{independent} caches does not restrict the caching gain attainable: independent caches are as powerful as fully randomized (or deterministic) caches. Finally, as \CG is NP-hard, so are all three other problems.

\section{Numerical Evaluation}\label{sec:numerical}
\begin{figure*}
\begin{center}
\hspace*{-1em}
\includegraphics[width=0.333333\textwidth]{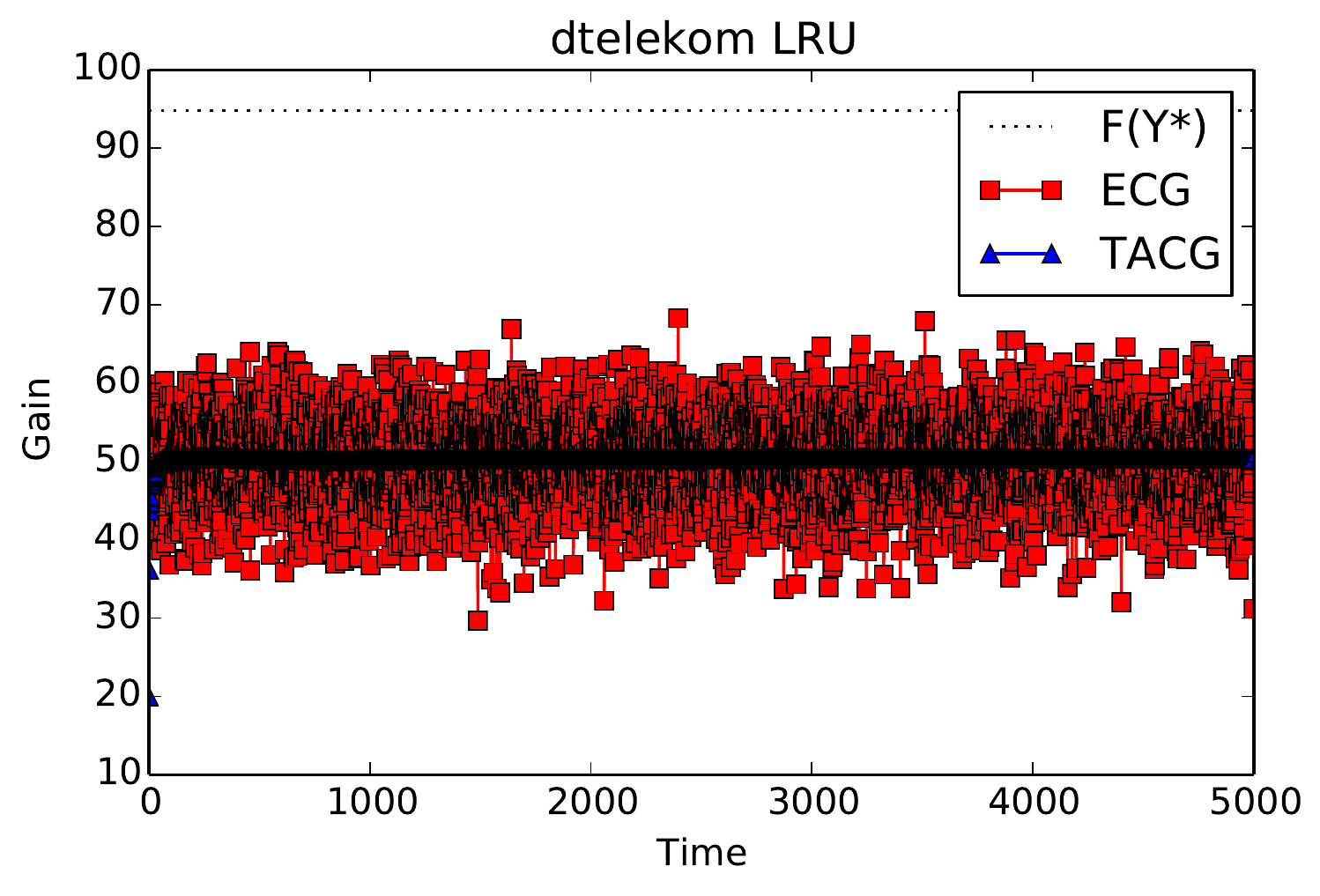}\hspace*{-1em}
\includegraphics[width=0.333333\textwidth]{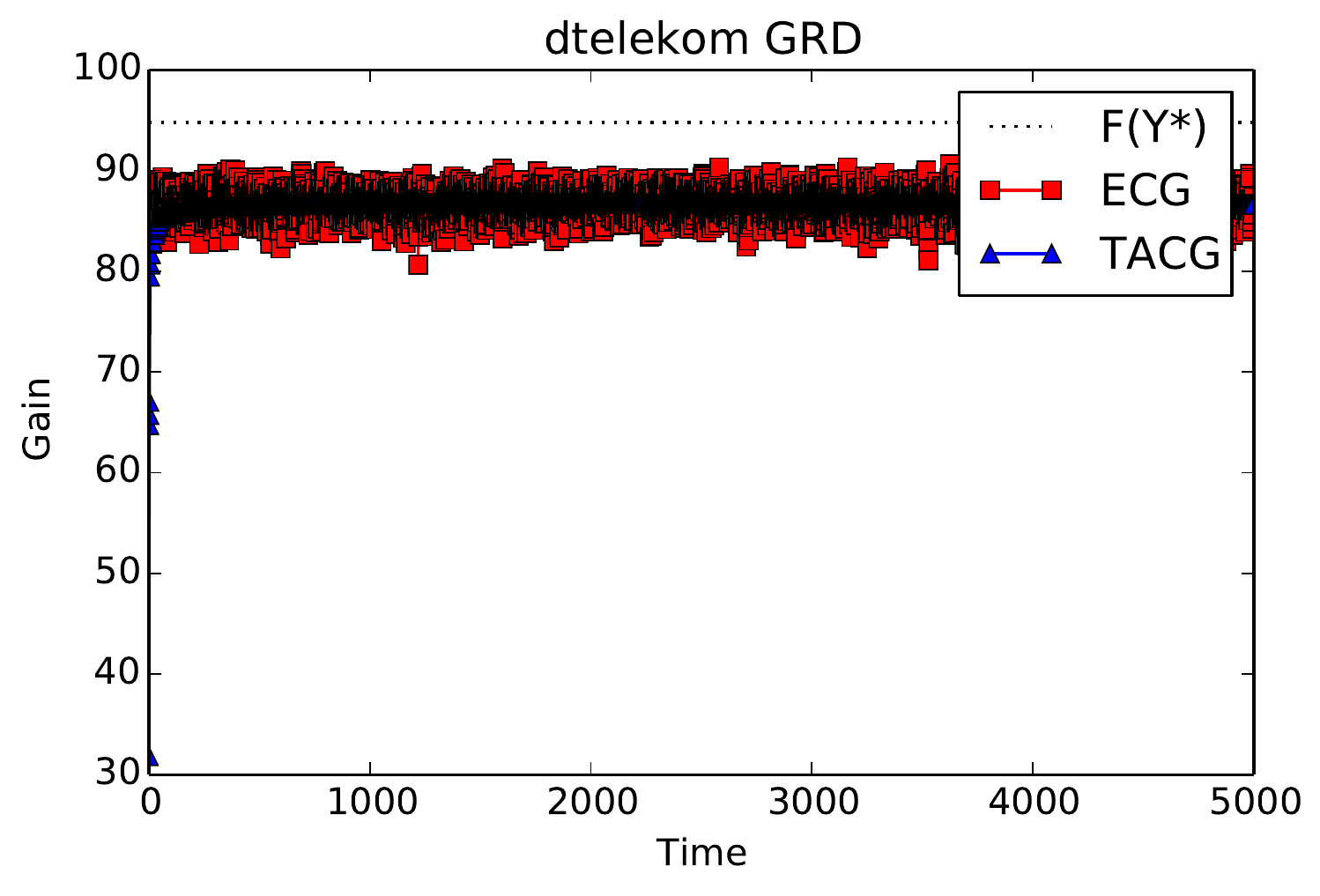}\hspace*{-1em}
\includegraphics[width=0.333333\textwidth]{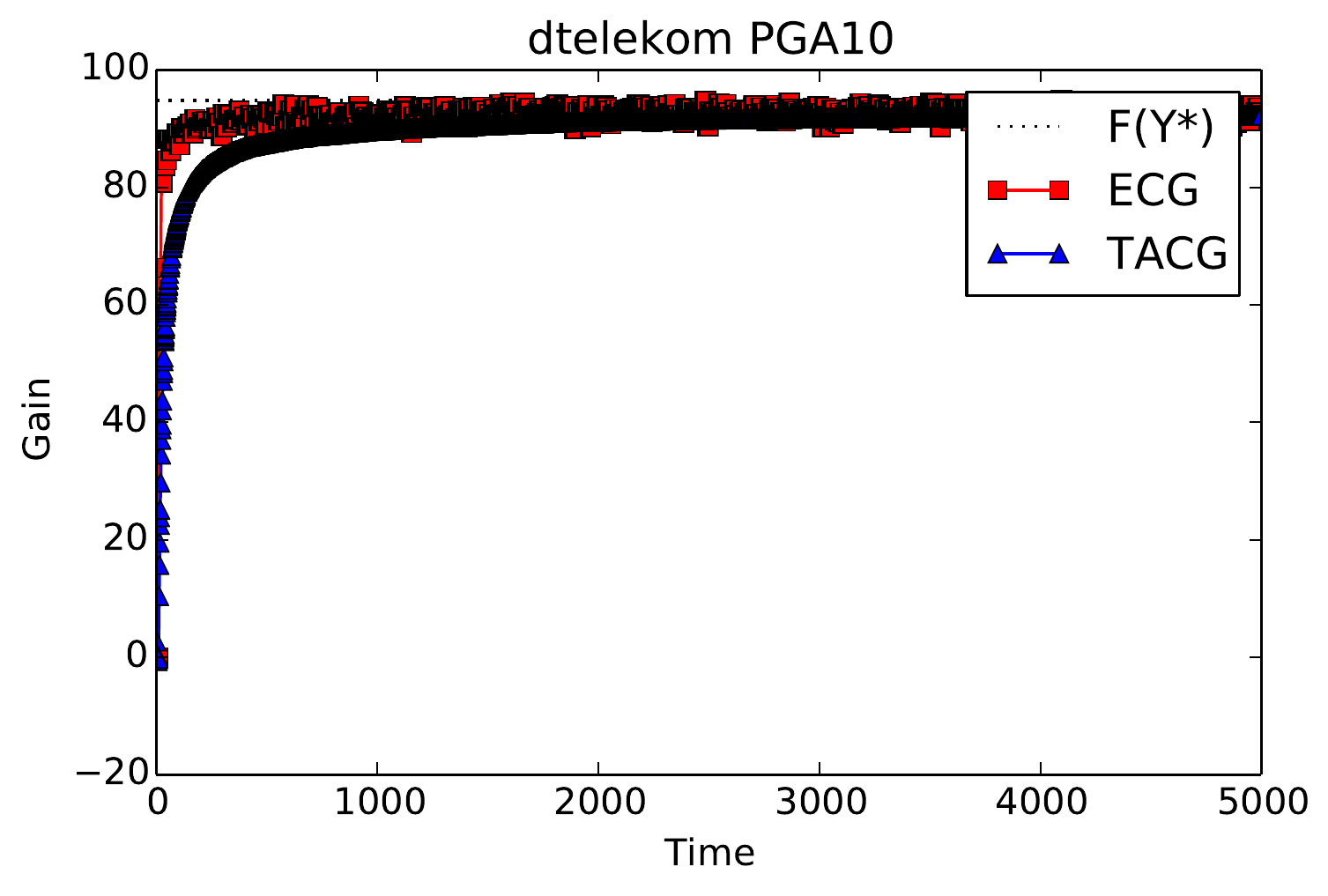}\hspace*{-1em}
\end{center}
\caption{Trajectories of the expected caching gain \texttt{ECG} and the time average caching gain \texttt{TACG} under the \texttt{LRU}, \texttt{GRD}, and \texttt{PGA} algorithms, the latter with $T=10$. The value $F(Y^*)$ for this experiment is shown in a dashed line. }\label{trajectories}
\end{figure*}

\begin{figure*}[!t]
\includegraphics[width=\textwidth]{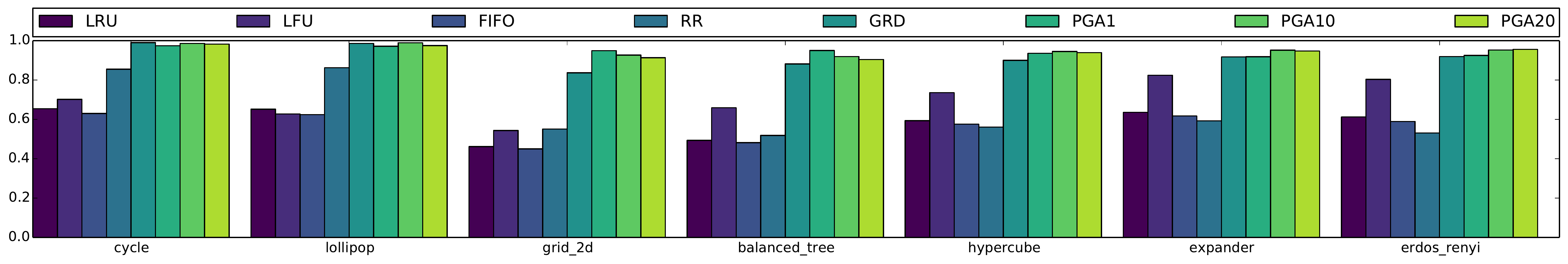}
\includegraphics[width=\textwidth]{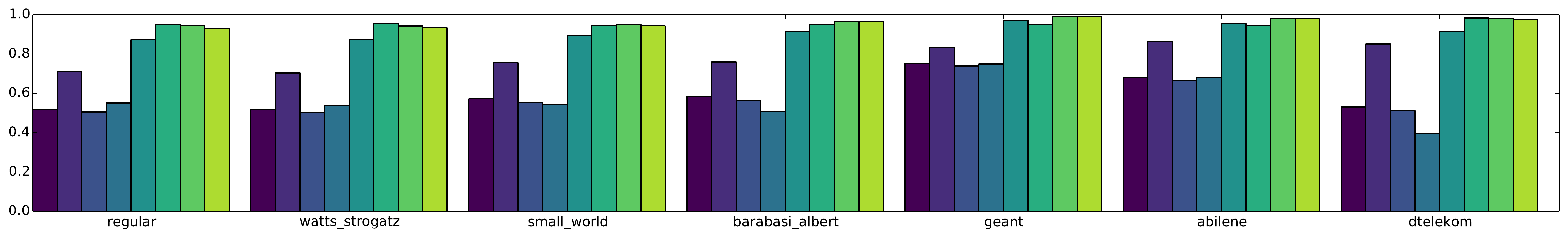}
\caption{Ratio of expected caching gain \texttt{ECG} to $F(Y^*)$, as given in Table~\ref{networks} under different networks and caching strategies. The greedy algorithm \texttt{GRD} performs almost as well as \texttt{PGA} in all cases. Both algorithms significantly outperform the remaining eviction policies.}\label{relativeperformance}
\end{figure*}

We simulate Algorithms~\ref{alg:ascent} and~\ref{alg:greedy} over synthetic and real networks, and compare their performance to path replication  combined with \texttt{LRU}, \texttt{LFU}, \texttt{FIFO}, and random replacement (\texttt{RR}) caches. Across the board,  greedy path replication performs exceptionally well, attaining at least 95\% of the expected caching gain attained by  Algorithm~\ref{alg:ascent}, while both significantly outperform traditional eviction policies.

\sloppy
\noindent\textbf{Topologies.} The networks we consider are summarized in Table~\ref{networks}. The first six graphs are deterministic. Graph \texttt{cycle} is a simple cyclic graph, and \texttt{lollipop} is a clique (i.e., complete graph), connected to a path graph of equal size. Graph \texttt{grid\_2d} is a two-dimensional square grid, \texttt{balanced\_tree} is a complete binary tree of depth 6, and \texttt{hypercube} is a 7-dimensional hypercube. Graph \texttt{expander} is a Margulies-Gabber-Galil expander \cite{gabber1981explicit}.
The next 5 graphs are random, i.e., were sampled from a probability distribution. Graph \texttt{erdos\_renyi} is an Erd\H{o}s-R\'enyi graph with parameter $p=0.1$, and \texttt{regular} is a 3-regular graph sampled u.a.r. The \texttt{watts\_strogatz} graph is a graph generated according to the Watts-Strogatz model of a small-world network \cite{watts1998collective} comrizing a cycle and $4$ randomly selected edges, while \texttt{small\_world} is the navigable small-world graph by Kleinberg \cite{kleinberg2000small}, comprising a grid with additional long range edges. The preferential attachment model of Barab\'asi and Albert \cite{barabasi1999emergence}, which yields powerlaw degrees, is used for \texttt{barabasi\_albert}.
Finally, the last 3 graphs represent the Deutche Telekom, Abilene, and GEANT backbone networks \cite{rossi2011caching}.

\fussy

\begin{table}[!t]
\begin{footnotesize}
\begin{tabular}{lp{1em}p{1em}p{1em}p{1em}p{1em}p{1em}c}
Graph & $|V|$ & $|E|$ & $|\catalog|$ & $|\requests|$ &  $|Q|$ & $c_v'$ & $F(Y^*)$ \\
\hline
\hline
\texttt{cycle}  & 30 & 60 & 10 & 100 & 10  & 2  &847.94       \\
 \texttt{lollipop} & 30 & 240 & 10 & 100 & 10 & 2 &735.78\\
\texttt{grid\_2d} &100 & 360  & 300 & 1K & 20 & 3& 381.09\\
\texttt{balanced\_tree} & 127 & 252  & 300 & 1K & 20 & 3& 487.39\\
\texttt{hypercube}&128 & 896  & 300 & 1K & 20 & 3& 186.29\\
 \texttt{expander} &100 & 716  & 300 & 1K & 20 & 3& 156.16\\
\hline
 \texttt{erdos\_renyi} &100 & 1042  & 300 & 1K & 20 & 3& 120.70\\
\texttt{regular} &100 & 300  & 300 & 1K & 20 & 3& 321.63 \\
 \texttt{watts\_strogatz} &100 & 400  & 300 & 1K & 20 & 3& 322.49  \\
 \texttt{small\_world} &100 & 491  & 300 & 1K & 20 & 3& 218.19 \\
 \texttt{barabasi\_albert} &100 & 768  & 300 & 1K & 20 & 3& 113.52 \\
\hline
 \texttt{geant} &22 & 66& 10&100 & 10 & 2 & 203.76\\
 \texttt{abilene} & 9 & 26 & 10 &100 & 10 & 2 & 121.25\\
 \texttt{dtelekom} & 68 & 546 & 300 & 1K & 20 & 3 &94.79\\
\end{tabular}
\end{footnotesize}
\caption{Graph Topologies and Experiment Parameters.}\label{networks}
\end{table}

\noindent\textbf{Experiment Setup.} We evaluate the performance of different adaptive strategies over the graphs in Table~\ref{networks}. Given a graph $G(V,E)$, we generate a catalog $\catalog$, and assign a cache to each node in the graph. For every item $i\in \catalog$, we designate a node  selected uniformly at random (u.a.r.) from $V$ as a source for this item. We set the capacity $c_v$ of every node $v$ so that $c_v'=c_v-|\{i:v\in S_i\}|$ is constant among all nodes in $V$. We assign a weight to each edge in $E$ selected u.a.r.~from the interval $[1,100]$.  We then generate a set of demands $\requests$ as follows. First, to ensure path overlaps, we select $|Q|$ nodes in $V$ u.a.r., that are the only nodes that generate requests; let $Q$ be the set of such query nodes. We generate a set of requests starting from a random node in $Q$, with the item requested selected from $\catalog$ according to a Zipf distribution with parameter $1.2$. The request is then routed over the shortest path between the node in $Q$ and the designated source for the requested item. We assign a rate $\lambda_{(i,p)}=1$ to every request in $\requests$, generated as above.

The values of $|\catalog|$, $|\requests|$ $|Q|$, and $c_v$ for each experiment are given in Table~\ref{networks}. For each experiment, we also provide in the last column the quantity
$F(Y^*),$ for $Y^* \in\argmax_{Y\in \feasibledomain_2} L(Y), $
i.e., the expected caching gain under a product form distribution that maximizes the relaxation $L$. Note that, by Theorem~\ref{approximation}, this is within a $1-1/e$ factor from the optimal expected caching gain.

\noindent\textbf{Caching Algorithms and Measurements.}
In all of the above networks, we evaluate the performance of Algorithms~\ref{alg:ascent} and~\ref{alg:greedy}, denoted by \texttt{PGA} (for Projected Gradient Ascent) and \texttt{GRD} (for Greedy), respectively. In the case of \texttt{PGA}, we tried different measurement periods $T=1.0,10.0,20.0$, termed \texttt{PGA1}, \texttt{PGA10}, and \texttt{PGA20}, respectively. We  implemented the algorithm both with state smoothening  \eqref{slide} and without (whereby allocations are sampled from marginals $Y^{(k)}$ directly). For brevity, we report only the non-smoothened versions, as time-average performance was nearly identical for both versions of the algorithm.

We also compare to path replication with \texttt{LRU}, \texttt{LFU}, \texttt{FIFO}, and \texttt{RR} eviction policies. In all cases, we simulate the network for 5000 time units. We collect measurements at epochs of a Poisson process with rate 1.0 (to leverage the PASTA property). In particular, at each measurement epoch, we extract the current allocation $X$, and compute the expected caching gain (\texttt{ECG}) as $F(X)$. In addition, we keep track of the actual cost of each request routed through the network, and compute the \text{time average caching gain} (\texttt{TACG}), measured as the difference of the cost of routing till the item source, minus the time average cost per request.

\noindent\textbf{Results.} Figure~\ref{trajectories} shows the trajectories of the expected caching gain \texttt{ECG} and the time average caching gain \texttt{TACG} under the \texttt{LRU}, \texttt{GRD}, and \texttt{PGA10} algorithms. All three algorithms converge relatively quickly to a steady state. \texttt{PGA10} converges the slowest but it indeed reaches $F(Y^*)$, as expected. In addition, the greedy heuristic \texttt{GRD} performs exceptionally well, converging very quickly (faster than \texttt{PGA}) to a value close to $F(Y^*)$. 
In contrast, the allocations reached in steady state by \texttt{LRU} are highly suboptimal, close to 50\% of $F(Y^*)$. Moreover, \texttt{LRU} exhibits high variability, spending considerable time in states with as low as 35\% of $F(Y^*)$. We note that the relatively low variability of both \texttt{GRD} and \texttt{PGA} is a desirable feature in practice, as relatively stable caches are preferred by network administrators.

The above observations hold across network topologies and caching algorithms. In Figure \ref{relativeperformance}, we plot the relative performance w.r.t \texttt{ECG} of all eight algorithms, normalized to $F(Y^*)$. We compute this as follows: to remove the effect of initial conditions, we focus on the interval $[1000,5000]$, and average the \texttt{ECG} values in this interval. 

We see that, in all cases, \texttt{PGA} attains $F(Y^*)$ for all three  values of the measurement period $T$. Moreover, the simple heuristic \texttt{GRD} has excellent performance: across the board, it attains more than 95\% of $F(Y^*)$, sometimes even outperforming \texttt{PGA}.
Both algorithms consistently outperform all other eviction policies. We observe that \texttt{RR} and \texttt{LFU} perform quite well in several cases, and that ``hard'' instances for one appear to be ``easy''  for the other. 

The differentiating instances, where performance is reversed, are the \texttt{cycle} and \texttt{lollipop} graphs: though small, these graphs contain long paths, in contrast to the remaining graphs that have a relatively low diameter. Intuitively, the long-path setting is precisely the scenario where local/myopic strategies like \texttt{LRU}, \texttt{LFU}, and \texttt{FIFO} make suboptimal decisions, while \texttt{RR}'s randomization helps.   
 
\section{Conclusions}\label{sec:conclusions}
\sloppy Our framework, and adaptive algorithm, leaves many open questions, including jointly optimizing caching \emph{and} routing decisions; this is even more pertinent in the presence of congestion, as in \cite{yeh2014vip}. The excellent performance of the greedy heuristic in our simulations further attests to the need for establishing its performance formally. Finally, studying adaptive caching schemes for fountain-coded content---which, by~\cite{shanmugam2013femtocaching}, has interesting connections to the relaxation $L$---is also an important open question.

\fussy

\section{Acknowledgements}
E. Yeh gratefully acknowledges support from National Science Foundation grant CNS-1423250 and a Cisco Systems research grant.

 \bibliographystyle{abbrv}

\appendix

\section{Proof of Theorem~{\protect \lowercase{\ref{approximation}}}}\label{proofofapprox}
To begin with, for all $Y\in \feasibledomain_2$, we have:
\begin{align} (1-\frac{1}{e}) L(Y)  \leq F(Y) \leq L(Y). \label{lem:sandwitch}\end{align}
To see this, note that
\begin{align*}
F(Y) &= \sum_{(i,p) \in \requests }\lambda_{(i,p)}  \sum_{k=1}^{|p|-1} w_{p_{k+1}p_{k}} \expect_{\nu}\left[\min\left\{1,\sum_{k'=1}^k x_{p_{k'}i}\right\}\right]\\
&\leq\sum_{(i,p) \in \requests }\lambda_{(i,p)}  \sum_{k=1}^{|p|-1} w_{p_{k+1}p_{k}} \min\left\{1,\sum_{k'=1}^k\expect_{\nu}\left[ x_{p_{k'}i}\right]\right\}
  \end{align*}
 by the concavity of the $\min$ operator, so
$F(Y) \leq L(Y).$
On the other hand, by Goemans and Williamson \cite{goemans1994new},
$$1-   \prod_{k'=1}^k (1-y_{p_{k'}i})\geq \left(1-(1 - 1/k)^k \right)\min\left\{1,\sum_{k'=1}^k y_{p_{k'}i}\right\}, $$
and the first inequality of the statement of the lemma follows as $(1-1/k)^k \leq \frac{1}{e}$.
By the optimality of $Y^*$ in $\feasibledomain_2$, clearly $F(Y^{**})\leq F(Y^*)$. By Lemma~\ref{lem:sandwitch} and the optimality of $Y^{**}$,
$F(Y^*)\leq  L(Y^*)\leq   L(Y^{**}) \leq \frac{e}{e-1}F(Y^{**}). $ \hspace*{\stretch{1}}\qed

\section{Proof of Lemma~{\protect \lowercase{\ref{placementcorrectness}}} }\label{appendix:corr}

We now prove the correctness of  Algorithm~\ref{alg:allocation}, by showing that   it produces a $\mu_v$ over $\bar{\feasibledomain}_1^v$ with marginals $\bar{y}_v$. For all $i\in \catalog$, denote by
\begin{align*}s_i &= \sum_{k=0}^{i-1}\bar{y}_{vk}, & t_i &= s_i+\bar{y}_{vi}, \tau_i&=t_i-\lfloor t_i \rfloor, \end{align*}
the quantities computed at Line \eqref{quantities} of the algorithm. Note that $s_0=0$,$t_{|\catalog|}=c_v$, and that $\tau_i$ is the fractional part of each $t_i$. Let also
$$0=\tau^{(0)}<\tau^{(1)}<\ldots<\tau^{(K)}=1$$
be the sequence of sorted $\tau_i$'s, with duplicates removed, as in Lines~\ref{sort}-\ref{sorted} of the algorithm. Note that, by construction, $K$ can be at most $|\catalog|$.

For $\ell\in \{0,\ldots,c-1\}$, $k\in\{0,\ldots,K-1\}$, let $A_\ell^k =(\ell+\tau^{(k)},\ell+\tau^{(k+1)}) $ be the open intervals in Line~\ref{find}. Then, the following lemma holds.
\begin{lemma}\label{uniqueness}
 For every $\ell\in\{1,\ldots,c-1\}$ and $k\in\{0,\ldots,K-1\}$, there exists exactly one $i\in \catalog$ s.t.  $(\ell+\tau^{(k)},\ell+\tau^{(k+1)})\subset [s_i,t_i]$. Moreover, any such $i$ must have $\bar{y}_{vi}>0$.
\end{lemma}
\begin{proof}
By construction $\tau^{(k)}<\tau^{(k+1)}$, so $A_\ell^k$ is a non-empty interval in $[0,c_v]$. For every $i\in \catalog$, $[s_i,t_i]$ are closed intervals (possibly of length zero, if $\bar{y}_{vi}=0$) such that $\bigcup_{i\in \catalog} [s_i,t_i]=[0,c_v]$. Hence, there must be at least one $[s_i,t_i]$ s.t. $A_\ell^k \cap [s_i,t_i] \neq \emptyset.$ To see that there can be no more than one, suppose that $A_\ell^k$ intersects more than one sets. Then it must intersect at least two consequtive sets, say $[s_j,t_j]$, $[s_{j+1},t_{j+1}]$ where, by construction $t_j=s_{j+1}$. This means that $t_j\in A_\ell^k$, and, which in turn implies that $\tau_j=t_j-\lfloor t_j\rfloor  \in (\tau^{(k)},\tau^{(k+1)})$, which is a contradiction, as the sequence $\tau^{(\cdot)}$ is sorted.

Hence, $A_\ell^k\cap[s_i,t_i]\neq \emptyset$ for exactly one $i\in \catalog$. For the same reason as above, $t_i$ cannot belong to $A_\ell^k$; if it did, then its fractional part $\tau_i$ would belong to $(\tau^{(k)},\tau^{(k+1)})$, a contradiction. Neither can $s_i$; if $s_i\in A_\ell^k$, then $s_i>0$, so $i>0$. This means that $s_i=t_{i-1}$, and if $t_{i-1}\in A_\ell^k$, we again reach a contradiction. Hence, the only way that $A_\ell^k\cap [s_i,t_i]\neq \emptyset$ is if $A_\ell^k\subset [s_i,t_i]$. Moreover, this implies that $s_i<=\tau^{(k)}$ and $\tau^{(k+1)}<=t_i$, which in turn implies that $\bar{y}_{vi}=t_i-s_i>0$, so the last statement of the lemma also follows.\hspace*{\stretch{1}}\qed
\end{proof}
The above lemma implies that Line \ref{find} of the algorithm always finds a unique $i\in \catalog$, for every $\ell\in\{1,\ldots,c_v-1\}$ and $k \in\{0,\ldots,K-1\}$. Denote this item by $i(k,\ell)$. The next lemma states that all such items obtained for different values of $\ell$ are distinct:
\begin{lemma}\label{span} For any two $\ell,\ell'\in \{1,\ldots,c_v-1\} $ with $\ell\neq \ell'$, $i(k,\ell)\neq i(k,\ell')$.
\end{lemma}
\begin{proof} Suppose that  $i(k,\ell) =i(k,\ell')$ for some $\ell'>\ell$. This means that both $A_\ell^k\subset [s_i,t_i]$ and $A_{\ell'}^k \subset [s_i,t_i]$. Thus,
\begin{align*}s_i &\leq \ell + \tau^{(k)} & &\text{and} & \ell'+\tau^{(k)} &<t_i . \end{align*} 
On the other hand, $\ell'>\ell$, so $\ell'>=\ell+1$. So the above imply that $\bar{y}_{vi}=t_i-s_i>\ell'-\ell>=1$, a contradiction, as any feasible $\bar{y}_{vi}$ must be at most 1. \hspace*{\stretch{1}}
\end{proof}
Observe that $\mu_v$ constructed by the algorithm is indeed a probability distribution, as (a) the differences $\tau^{(k+1)}-\tau^{(k)}$ are, by construction, positive, and (b) their sum is 1. Moreover, the above two lemmas imply that the vectors $x$ constructed by the Algorithm contain exactly $c_v$ non-zero elements, so $\mu_v$ is indeed a distribution over $\feasibledomain_1$. The last lemma establishes that the constructed distribution has the desirable marginals, thereby completing the proof of correctness.
\begin{lemma}\label{marginallemma} For any $i\in \catalog$ such that $\bar{y}_{vi}>0$, $$\sum_{x\in \supp(\mu_v): x_i=1} \mu_v(x) = \bar{y}_{vi} .$$
\end{lemma}
\begin{proof}
Note that the sets $A_{\ell}^k$ are disjoint,  have non-empty interior, and their union $\mathcal{U}=\bigcup_{\ell,k} A_{\ell}^k$ has Borel measure (i.e., length) $c_v$. Consider an $i\in \catalog$ s.t. $\bar{y}_{vi}>0$. Then, there must be a set $A_{\ell}^k$ that intersects $[s_i,t_i]$; if not, then the union  $\mathcal{U}$ would have Borel measure at most $1-\bar{y}_{vi}$, a contradiction. As in the proof of Lemma~\eqref{uniqueness}, the set  $A_{\ell}^k$ that intersects $[s_i,t_j]$ must be a subset of $[s_i,t_j]$. Consider the remainder, i.e., the set difference $[s_i,t_i]\setminus A_{\ell}^k$. By the same argument as above, if this remainder has non-zero Borel measure, there must be an interval $A_{\ell'}^{k'}\in \mathcal{U}$, different from  $A_{\ell}^k$, that intersects it. As before, this set must be included in $[s_i,t_i]$, and as it is disjoint from $A_{\ell}^k$, it will be included in the remainder. We can therefore construct a sequence of such sets  $A_{\ell}^k\in \mathcal{U}$ that are included in $[s_i,t_i]$, so long as their remainder has non-zero Borel measure. Since there are finitely many such sets, this sequence will be finite, and when it terminates the remainder will have Borel measure zero. Hence, the union of these disjoint sets has Borel measure exactly $\bar{y}_{vi}$.

Every interval in this sequence corresponds to an allocation $x$ constructed by the algorithm s.t.~$x_i=1$. By Lemma~\ref{span}, each interval corresponds to a distinct such allocation. Moreover, since the remainder of this construction has Borel measure zero, no other set in $\mathcal{U}$ can intersect $[s_i,t_i]$. Finally, the probability of these allocations under $\mu_v$ is equal to the sum of the Borel measures of this sets; as they are disjoint, the latter is equal to the Borel measure of their union, which is $\bar{y}_{vi}$.\hspace*{\stretch{1}}\qed
\end{proof}

 \end{document}